\documentclass{amsart}%
\usepackage{amsfonts}
\usepackage{amsmath}
\usepackage{amssymb}
\usepackage{graphicx}
\usepackage{subfigure}%
\graphicspath{ {celebro/},{Image with noise/},{Images without noise/},{Lena with noise/}}
\newtheorem{theorem}{Theorem}
\theoremstyle{plain}

\newtheorem{definition}{Definition}

\newtheorem{lemma}{Lemma}

\newtheorem{proposition}{Proposition}
\newtheorem{remark}{Remark}

\numberwithin{equation}{section}
\begin{document}
\title[$p$-adic Continuous CNNs]{$p$-adic Cellular Neural Networks}
\author[Zambrano-Luna]{B. A. Zambrano-Luna}
\address{Centro de Investigaci\'{o}n y de Estudios Avanzados del Instituto
Polit\'{e}cnico Nacional\\
Departamento de Matem\'{a}ticas, Av. Instituto Polit\'{e}cnico Nacional \#
2508, Col. San Pedro Zacatenco, CDMX. CP 07360 \\
M\'{e}xico.}
\email{bazambrano@math.cinvestav.mx}
\author[Z\'{u}\~{n}iga-Galindo]{W. A. Z\'{u}\~{n}iga-Galindo}
\address{University of Texas Rio Grande Valley\\
School of Mathematical \& Statistical Sciences\\
One West University Blvd\\
Brownsville, TX 78520, United States and Centro de Investigaci\'{o}n y de
Estudios Avanzados del Instituto Polit\'{e}cnico Nacional\\
Departamento de Matem\'{a}ticas, Unidad Quer\'{e}taro\\
Libramiento Norponiente \#2000, Fracc. Real de Juriquilla. Santiago de
Quer\'{e}taro, Qro. 76230\\
M\'{e}xico.}
\email{wilson.zunigagalindo@utrgv.edu, wazuniga@math.cinvestav.edu.mx}
\thanks{The second author was partially supported by Conacyt Grant No. 250845 (Mexico)
and by the Debnath Endowed Professorship (UTRGV, USA)}
\keywords{Cellular neural networks, Hierarchies, Deep learning, $p$-Adic numbers.}

\begin{abstract}
In this article we introduce the $p$-adic cellular neural networks which are
mathematical generalizations of the classical cellular neural networks (CNNs)
introduced by Chua and Yang. The new networks have infinitely many cells which
are organized hierarchically in rooted trees, and also they have infinitely
many hidden layers. Intuitively, the $p$-adic CNNs occur as limits of large
hierarchical discrete CNNs. More \ precisely, the new networks can be very
well approximated by hierarchical discrete CNNs. Mathematically speaking, each
of the new networks is modeled by one integro-differential equation depending
on several $p$-adic spatial variables and the time. We study the Cauchy
problem associated to these integro-differential equations and also provide
numerical methods for solving them.

\end{abstract}
\maketitle

\section{Introduction}

In the late 80s Chua and Yang introduced a new natural computing paradigm
called the cellular neural networks (or cellular nonlinear networks) CNN which
includes the cellular automata as a particular case \cite{Chua-Yang},
\cite{Chua-Yang2}, \cite{Chua}. From the beginning the CNN paradigm was
intended for applications as an integrated circuit. This paradigm has been
extremely successful in various applications in vision, robotics and remote
sensing, see e.g. \cite{Chua-Tamas}, \cite{Slavova} and the references therein.

In this article we present a mathematical generalization of the CNNs of Chua
and Yang called $p$\textit{-adic cellular neural networks}. The $p$-adic
continuous CNNs offer a theoretical framework to study the emergent patterns
of hierarchical discrete CNNs having arbitrary many hidden layers.

Nowadays, it is widely accepted that the analysis on ultrametric spaces is the
natural tool for formulating models where the hierarchy plays a central role.
An ultrametric space $(M,d)$ is a metric space $M$ with a distance satisfying
$d(A,B)\leq\max\left\{  d\left(  A,C\right)  ,d\left(  B,C\right)  \right\}  $
for any three points $A$, $B$, $C$ in $M$. Ultrametricity in physics means the
emergence of ultrametric spaces in physical models. Ultrametricity was
discovered in the 80s by Parisi and others in the theory of spin glasses and
by Frauenfelder and others in physics of proteins. In both cases, the space of
states of a complex system has a hierarchical structure which play a central
role in the physical behavior of the system, see e.g. \cite{Dra-Kh-K-V},
\cite{Fraunfelder et al}, \cite{Hua et al}, \cite{KKZuniga}, \cite{Kozyrev
SV}, \cite{R-T-V}, \cite{V-V-Z}, \cite{Zuniga-JMAA-2020}%
-\cite{Zuniga-Nonlinearity}, and the references therein.

On the other hand, Khrennikov and his collaborators have studied neural
network models where $p$-state neurons take their values in $p$-adic numbers,
see \cite{Albeverio-Khrennikov-Tirozzi}, \cite{Krennikov-tirozzi}. These
models are completely different to the ones considered here. In addition,
Khrennikov has developed non-Archimedean models of brain activity and mental
processes, see e.g. \cite{Khrenikov2} and the references therein.

Among the ultrametric spaces, the field of $p$-adic numbers $\mathbb{Q}_{p}$
plays a central role. A $p$-adic number is a series of the form%
\begin{equation}
x=x_{-k}p^{-k}+x_{-k+1}p^{-k+1}+\ldots+x_{0}+x_{1}p+\ldots,\text{ with }%
x_{-k}\neq0\text{,} \label{p-adic numbers}%
\end{equation}
where $p$ is a prime number, the $x_{j}$s \ are $p$-adic digits, i.e. numbers
in the set $\left\{  0,1,\ldots,p-1\right\}  $. The set of all the possible
series of form (\ref{p-adic numbers}) constitutes the field of $p$-adic
numbers $\mathbb{Q}_{p}$. There are natural field operations, sum and
multiplication, on series of form (\ref{p-adic numbers}), see e.g.
\cite{Koblitz}. There is also a natural norm in $\mathbb{Q}_{p}$ defined as
$\left\vert x\right\vert _{p}=p^{k}$, for a nonzero $p$-adic number $x$ of the
form (\ref{p-adic numbers}). The field of $p$-adic numbers with the distance
induced by $\left\vert \cdot\right\vert _{p}$ is a complete ultrametric space.
The ultrametric property refers to the fact that $\left\vert x-y\right\vert
_{p}\leq\max\left\{  \left\vert x-z\right\vert _{p},\left\vert z-y\right\vert
_{p}\right\}  $ for any $x$, $y$, $z$ in $\mathbb{Q}_{p}$.

We denote by $G_{M}$ the set all the $p$-adic numbers of the form
$\boldsymbol{i}=\boldsymbol{i}_{-M}p^{-M}+\boldsymbol{i}_{-M+1}p^{-M+1}%
+\cdots+\boldsymbol{i}_{0}+\cdots+\boldsymbol{i}_{M-1}p^{M-1}$, where the
$\boldsymbol{i}_{j}$s belong to $\left\{  0,1,\ldots,p-1\right\}  $. Then
($G_{M},\left\vert \cdot\right\vert _{p}$) is a finite ultrametric space.
Geometrically speaking, $G_{M}$ is a regular rooted tree with $2M$ layers,
here regular means that exactly $p$ edges emanate from each vertex. A
($1$-dimensional) $p$-adic discrete CNN is a dynamical system of the form%
\begin{equation}
\frac{\partial}{\partial t}X(\boldsymbol{i},t)=-X(\boldsymbol{i},t)+%
{\displaystyle\sum\limits_{\boldsymbol{j}\in G_{M}}}
\mathbb{A}(\boldsymbol{i},\boldsymbol{j})Y(\boldsymbol{j},t)+%
{\displaystyle\sum\limits_{\boldsymbol{j}\in G_{M}}}
\mathbb{B}(\boldsymbol{i},\boldsymbol{j})U(\boldsymbol{j})+Z(\boldsymbol{i}),
\label{Network_1}%
\end{equation}
$\boldsymbol{i}\in G_{M}$, where $Y(\boldsymbol{j},t)=f\left(
X(\boldsymbol{j},t)\right)  $, with $f(x)=\frac{1}{2}\left(
|x+1|-|x-1|\right)  $. Here $X({\boldsymbol{i}},t),Y({\boldsymbol{i}}%
,t)\in\mathbb{R}$ are the state, respectively the output, of cell
$\boldsymbol{i}$ at the time $t$. The function $U(\boldsymbol{i})\in
\mathbb{R}$ is the input of the cell $\boldsymbol{i}$, $Z(\boldsymbol{i}%
)\in\mathbb{R}$ is the threshold of cell $\boldsymbol{i}$, and the matrices
$\mathbb{A},\mathbb{B}:G_{M}^{N}\times G_{M}^{N}\rightarrow\mathbb{R}$ are the
feedback operator and feedforward operator, respectively. Notice that matrices
$\mathbb{A}$, $\mathbb{B}$ are functions on the Cartesian product of two
rooted trees. The Chua-Yang CNNs are a particular case of (\ref{Network_1}).
In this article we study $N$-dimensional, discrete hierarchical CNNs having
arbitrary many layers. For the seek of simplicity, we focus on space-invariant
networks, i.e. in the case in which
\begin{equation}
\mathbb{A}(\boldsymbol{i},\boldsymbol{j})=\mathbb{A}(\left\vert \boldsymbol{i}%
-\boldsymbol{j}\right\vert _{p})\text{, \ \ }\mathbb{B}(\boldsymbol{i}%
,\boldsymbol{j})=\mathbb{B}(\left\vert \boldsymbol{i}-\boldsymbol{j}%
\right\vert _{p}). \label{Condition}%
\end{equation}
In this article we initiate the study of the emergent patterns produced by the
$p$-adic discrete CNNs. Since we are interested in arbitrary large trees, the
description of these networks requires literally of millions of
integro-differential equations, consequently a numerical approach seems not
suitable, instead of this, we construct a $p$-adic continuous model that can
be very well approximated by (\ref{Network_1}).

The study of the qualitative behavior of differential equations on large
graphs is a relevant matter due its applications. In \cite{Nakao-Mikhailov}
Nakao and Mikhailov proposed using continuous models to study
reaction-diffusion systems on networks and the corresponding Turing patterns.
In \cite{Zuniga-JMAA-2020} the second author showed that $p$-adic analysis is
the natural tool to carry out this program. Models constructed using energy
landscapes naturally drive to a large systems of differential equations (the
master equation of the system), see e.g. \cite{Becker et al}, \cite{KKZuniga},
\cite{Kozyrev SV}. $p$-Adic continuous versions of some of these systems were
constructed by Avetisov, Kozyrev and others in connection with models of
protein folding, see e.g. \cite{KKZuniga}, \cite{Kozyrev SV}\ for a general
discussion. Another relevant system is the Eigen-Schuster model in biology. In
\cite{Zuniga-JPA-208} a $p$-adic continuous version of this model was
introduced, this $p$-adic version allows to explain the Eigen paradox.
Recently Hua and Hovestadt pointed out that the $p$-adic number system offers
a natural representation of hierarchical organization of complex networks
\cite{Hua et al}.

Intuitively, in the space-invariant case, the continuous model corresponding
to (\ref{Network_1}) is obtained by taking the limit as $M$ tends to infinity:%
\begin{align}
\frac{\partial X(x,t)}{\partial t}  &  =-X(x,t)+%
{\displaystyle\int\limits_{\mathbb{Q}_{p}}}
A(\left\vert x-y\right\vert _{p})Y(y,t)dy+\label{Network_2A}\\
&
{\displaystyle\int\limits_{\mathbb{Q}_{p}}}
B(\left\vert x-y\right\vert _{p})U(y)dy+Z(x),\nonumber
\end{align}
with $Y(x,t)=f(X(x,t))$. For the sake of simplicity, in the introduction we
discuss our results in dimension one. We study the case where $A(\left\vert
x\right\vert _{p})$, $B(\left\vert x\right\vert _{p})$ are integrable, and
$U$, $Z$ are continuous functions vanishing at infinity. Under these
hypotheses the initial value problem attached to (\ref{Network_2A}), with
\ initial datum $X_{0}$ (a continuous function vanishing at infinity) \ has a
unique solution $X(x,t)$ which is a continuos function vanishing at infinity
in $x$ for any $t\geq0$, satisfying $\left\vert X(x,t)\right\vert \leq
X_{\max}$, where the constant $X_{\max}$ is completely determined by $A$, $B$,
$U$, $Z$ and $f$, see Theorem \ref{Theorem1}. \ An analog result is valid for
discrete CNNs, see Theorem \ref{Theorem1A}.

The solution $X(x,t)$ can be very well approximated in the $\left\Vert
\cdot\right\Vert _{\infty}$-norm as
\[
\sum_{\boldsymbol{j}\in G_{M}}X(\boldsymbol{j},t)\Omega\left(  p^{M}\left\vert
x-\boldsymbol{j}\right\vert _{p}\right)  .
\]
By using standard techniques of approximation of semilinear evolution
equations, we show that the solution of the Cauchy problem attached to
(\ref{Network_1}), under condition (\ref{Condition}), is arbitrarily closed in
the $\left\Vert \cdot\right\Vert _{\infty}$-norm to the solution of the Cauchy
problem attached to (\ref{Network_2A}), if $M$ is sufficiently large, see
Theorem \ref{Theorem2}. This implies that the $p$-adic continuous CNNs have
infinitely many hidden layers, and that they are continuous versions of
suitable $p$-adic discrete CNNs. It is relevant to mention that equation
(\ref{Network_2A}) makes sense over the real numbers, i.e. by replacing
$\mathbb{Q}_{p}$ by $\mathbb{R}$ in (\ref{Network_2A}) we get an equation
modeling a continuous network. But, there are no natural discretizations of
the real version of (\ref{Network_2A}) that can be interpreted as hierarchical
CNNs, because the real numbers are a completely ordered field, and thus the
natural hierarchy is only the linear one.

In practical applications it is natural to assume that radial functions $A$,
$B$ have compact support or that they are test functions. Under this
hypothesis we study the patterns produced by $p$-adic continuous CNNs when
$U$, $Z$ and $X_{0}$ are test functions. The hypothesis that $X_{0}$ is a test
functions means that at time $t=0$ only certain clusters of cells are excited.
Each cluster corresponds to a $p$-adic ball \ centered at some cell with
radius, say $p^{-L}$. The intensity of the excitation is the same for all
cells in a given cluster. The fact that $U$, $Z$ are test functions can be
interpreted in an analogous way. Let $B_{M_{0}}$ denote the ball centered at
the origin with radius $p^{M_{0}}$, which the smallest ball containing the
supports of $A$, $B$, $U$, $Z$, $X_{0}$. Then the solution $X(x,t)$ of the
initial value problem attached to (\ref{Network_2A}) is a test function
supported in $B_{M_{0}}$ of the form $\sum_{\boldsymbol{j}\in G_{M_{0}}%
}X(\boldsymbol{j,t})\Omega\left(  p^{M_{0}}\left\vert x-\boldsymbol{j}%
\right\vert _{p}\right)  $ for $t\geq0$, with $M_{0}\geq L$, see Theorem
\ref{Theorem0}. This means that a $p$-adic continuous CNN produces a pattern
which is organized in a finite number of \ disjoint clusters, each of them
supporting a time varying pattern. We also show the existence of two steady
state patterns $X_{+}(x)$, $X_{-}(x)$, which are test functions, such that
$\ X_{-}(x)\leq\lim_{t\rightarrow\infty}X(x,t)\leq X_{+}(x)$, see Theorem
\ref{Theorem1}. We conjecture that for generic $p$-adic continuous CNNs,
$\lim_{t\rightarrow\infty}X(x,t)$ is a test function, which means that the
steady state pattern is organized in a finite number of \ disjoint clusters,
each of them supporting a constant pattern. This is exactly the multistability
property reported in \cite{Nakao-Mikhailov}, see also \cite{Zuniga-JMAA-2020},
for reaction-diffusion networks.

We have conducted a large number of numerical simulations. Such simulations
require solving integro-differential equations on a tree. The numerical study
of $p$-adic continuous CNNs offers two big challenges. The first, the need of
dealing with matrices having millions of entries, the second, the
visualization of functions depending on $p$-adic variables. Due to the first
problem, we use small trees with $16$ to $64$ leaves. The $p$-adic numbers
have a fractal nature, then, it is necessary to visualize real-valuated functions defined on the Cartesian product of a fractal times the real line. 
To deal with this problem we us systematically heat maps which allow us to get a 
glimpse of the hierarchical nature of the CNNs. Our numerical simulations show
that the solutions of continuous CNNs exhibit a very complex behavior,
including self-similarity and multistability, depending on the interaction of
all the parameters defining the network \ and initial datum.

\section{$p$\textbf{-}Adic Analysis: Essential Ideas}

\subsection{The field of $p$-adic numbers}

Along this article $p$ will denote a prime number. The field of $p-$adic
numbers $\mathbb{Q}_{p}$ is defined as the completion of the field of rational
numbers $\mathbb{Q}$ with respect to the $p-$adic norm $|\cdot|_{p}$, which is
defined as%
\[
\left\vert x\right\vert _{p}=\left\{
\begin{array}
[c]{lll}%
0 & \text{if} & x=0\\
p^{-\gamma} & \text{if} & x=p^{\gamma}\frac{a}{b}\text{,}%
\end{array}
\right.
\]
where $a$ and $b$ are integers coprime with $p$. The integer $\gamma:=ord(x)$,
with $ord(0):=+\infty$, is called the\textit{\ }$p-$\textit{adic order of}
$x$. We extend the $p-$adic norm to $\mathbb{Q}_{p}^{N}$ by taking%
\[
||x||_{p}:=\max_{1\leq i\leq N}|x_{i}|_{p},\qquad\text{for }x=(x_{1}%
,\dots,x_{N})\in\mathbb{Q}_{p}^{N}.
\]
We define $ord(x)=\min_{1\leq i\leq N}\{ord(x_{i})\}$, then $||x||_{p}%
=p^{-ord(x)}$.\ The metric space $\left(  \mathbb{Q}_{p}^{N},||\cdot
||_{p}\right)  $ is a complete ultrametric space. As a topological space
$\mathbb{Q}_{p}$\ is homeomorphic to a Cantor-like subset of the real line,
see e.g. \cite{Alberio et al}, \cite{V-V-Z}.

Any $p-$adic number $x\neq0$ has a unique expansion of the form
\[
x=p^{ord(x)}\sum_{j=0}^{\infty}x_{j}p^{j},
\]
where $x_{j}\in\{0,1,\dots,p-1\}$ and $x_{0}\neq0$.

\subsection{Topology of $\mathbb{Q}_{p}^{N}$}

For $r\in\mathbb{Z}$, denote by $B_{r}^{N}(a)=\{x\in\mathbb{Q}_{p}%
^{N};||x-a||_{p}\leq p^{r}\}$ \textit{the ball of radius }$p^{r}$ \textit{with
center at} $a=(a_{1},\dots,a_{N})\in\mathbb{Q}_{p}^{N}$, and take $B_{r}%
^{N}(0):=B_{r}^{N}$. Note that $B_{r}^{N}(a)=B_{r}(a_{1})\times\cdots\times
B_{r}(a_{N})$, where $B_{r}(a_{i}):=B_{r}^{1}(a_{i})=\{x\in\mathbb{Q}%
_{p};|x_{i}-a_{i}|_{p}\leq p^{r}\}$ is the one-dimensional ball of radius
$p^{r}$ with center at $a_{i}\in\mathbb{Q}_{p}$. The ball $B_{0}^{N}$ equals
the product of $N$ copies of $B_{0}=\mathbb{Z}_{p}$, \textit{the ring of }%
$p-$\textit{adic integers}. We also denote by $S_{r}^{N}(a)=\{x\in
\mathbb{Q}_{p}^{N};||x-a||_{p}=p^{r}\}$ \textit{the sphere of radius }$p^{r}$
\textit{with center at} $a=(a_{1},\dots,a_{N})\in\mathbb{Q}_{p}^{N}$, and take
$S_{r}^{N}(0):=S_{r}^{N}$. We notice that $S_{0}^{1}=\mathbb{Z}_{p}^{\times}$
(the group of units of $\mathbb{Z}_{p}$), but $\left(  \mathbb{Z}_{p}^{\times
}\right)  ^{N}\subsetneq S_{0}^{N}$. The balls and spheres are both open and
closed subsets in $\mathbb{Q}_{p}^{N}$. In addition, two balls in
$\mathbb{Q}_{p}^{N}$ are either disjoint or one is contained in the other.

As a topological space $\left(  \mathbb{Q}_{p}^{N},||\cdot||_{p}\right)  $ is
totally disconnected, i.e. the only connected \ subsets of $\mathbb{Q}_{p}%
^{N}$ are the empty set and the points. A subset of $\mathbb{Q}_{p}^{N}$ is
compact if and only if it is closed and bounded in $\mathbb{Q}_{p}^{N}$, see
e.g. \cite[Section 1.3]{V-V-Z}, or \cite[Section 1.8]{Alberio et al}. The
balls and spheres are compact subsets. Thus $\left(  \mathbb{Q}_{p}%
^{N},||\cdot||_{p}\right)  $ is a locally compact topological space.

We will use $\Omega\left(  p^{-r}||x-a||_{p}\right)  $ to denote the
characteristic function of the ball $B_{r}^{N}(a)$. For more general sets, we
will use the notation $1_{A}$ for the characteristic function of a set $A$.

\subsection{The Bruhat-Schwartz space}

A real-valued function $\varphi$ defined on $\mathbb{Q}_{p}^{N}$ is
\textit{called locally constant} if for any $x\in\mathbb{Q}_{p}^{N}$ there
exists an integer $l(x)\in\mathbb{Z}$ such that%
\begin{equation}
\varphi(x+x^{\prime})=\varphi(x)\text{ for }x^{\prime}\in B_{l(x)}^{N}.
\label{local_constancy_0}%
\end{equation}

A function $\varphi:\mathbb{Q}_{p}^{N}\rightarrow\mathbb{R}$ is called a
\textit{Bruhat-Schwartz function (or a test function)} if it is locally
constant with compact support. Any test function can be represented as a
linear combination, with real coefficients, of characteristic functions of
balls. The $\mathbb{R}$-vector space of Bruhat-Schwartz functions is denoted
by $\mathcal{D}(\mathbb{Q}_{p}^{N})$. For $\varphi\in\mathcal{D}%
(\mathbb{Q}_{p}^{N})$, the largest number $l=l(\varphi)$ satisfying
(\ref{local_constancy_0}) is called \textit{the exponent of local constancy
(or the parameter of constancy) of} $\varphi$.

If $U$ is an open subset of $\mathbb{Q}_{p}^{N}$, $\mathcal{D}(U)$ denotes the
space of test functions with supports contained in $U$, then $\mathcal{D}(U)$
is dense in
\[
L^{\rho}\left(  U\right)  =\left\{  \varphi:U\rightarrow\mathbb{R};\left(
{\displaystyle\int\limits_{\mathbb{Q}_{p}^{N}}}
\left\vert \varphi\left(  x\right)  \right\vert ^{\rho}d^{N}x\right)
^{\frac{1}{\rho}}<\infty\right\}  ,
\]
where $d^{N}x$ is the Haar measure on $\mathbb{Q}_{p}^{N}$ normalized by the
condition $vol(B_{0}^{N})\allowbreak=1$, for $1\leq\rho<\infty$, see e.g.
\cite[Section 4.3]{Alberio et al}. In the case $U=\mathbb{Q}_{p}^{N}$, we will
use the notation $L^{\rho}$ instead of $L^{\rho}\left(  \mathbb{Q}_{p}%
^{N}\right)  $. For an in depth discussion about $p$-adic analysis the reader
may consult \ \cite{Alberio et al}, \cite{Koch}, \cite{Taibleson},
\cite{V-V-Z}.

\subsection{The Spaces $\mathcal{X}_{\infty}$, $\mathcal{X}_{M}$}

We define $\mathcal{X}_{\infty}(%
\mathbb{Q}
_{p}^{N}):=\mathcal{X}_{\infty}=\overline{\left(  \mathcal{D}(%
\mathbb{Q}
_{p}^{N}),\left\Vert \cdot\right\Vert _{\infty}\right)  }$, where $\left\Vert
\phi\right\Vert _{\infty}=\sup_{x\in%
\mathbb{Q}
_{p}^{N}}|\phi(x)|$ and the bar means the completion with respect the metric
induced by $\left\Vert \cdot\right\Vert _{\infty}$. Notice that all the
functions in $\mathcal{X}_{\infty}$ are continuous and \ that
\[
\mathcal{X}_{\infty}\subset\mathcal{C}_{0}:=\left(  \left\{  f:%
\mathbb{Q}
_{p}^{N}\rightarrow\mathbb{R};f\text{ continuous with }\lim_{\left\Vert
x\right\Vert _{p}\rightarrow\infty}f\left(  x\right)  =0\right\}  ,\left\Vert
\cdot\right\Vert _{\infty}\right)  .
\]
On the other hand, since $\mathcal{D}(%
\mathbb{Q}
_{p}^{N})$ is dense in $\mathcal{C}_{0}$, cf. \cite[Chap. II, Proposition
1.3]{Taibleson}, we conclude that $\mathcal{X}_{\infty}=\mathcal{C}_{0}$.

For $M\geq1$, we set $G_{M}^{N}:=B_{M}^{N}/B_{-M}^{N}$, which is a finite
additive group with $\#G_{M}^{N}:=p^{2NM}$ elements. Any element
$\boldsymbol{i}=(\boldsymbol{i}_{1},\ldots,\boldsymbol{i}_{N})$ of $G_{M}^{N}$
can be represented as
\begin{equation}
\boldsymbol{i}_{j}=\boldsymbol{i}_{-M}^{j}p^{-M}+\boldsymbol{i}_{-M+1}%
^{j}p^{-M+1}+\ldots+\boldsymbol{i}_{0}^{j}+\boldsymbol{i}_{1}^{j}%
p+\ldots+\boldsymbol{i}_{M-1}^{j}p^{M-1}\text{,} \label{representatives}%
\end{equation}
for $j=1,\ldots,N$, with $\boldsymbol{i}_{k}^{j}\in\left\{  0,1,\ldots
,p-1\right\}  $. From now on, we fix a set of representatives in $%
\mathbb{Q}
_{p}^{N}$ for $G_{M}^{N}$ of the form (\ref{representatives}). Notice that
\[
\boldsymbol{i}_{j}=p^{-M}\left(  \boldsymbol{a}_{0}^{j}+\boldsymbol{a}_{1}%
^{j}p+\cdots+\boldsymbol{a}_{2M-1}^{j}p^{2M-1}\right)  ,
\]
where $\boldsymbol{a}_{0}^{j}+\boldsymbol{a}_{1}^{j}p+\cdots+\boldsymbol{a}%
_{2M-1}^{j}p^{2M-1}\in\mathbb{Z}_{p}/p^{2M}\mathbb{Z}_{p}=B_{0}/B_{-2M}$.

The functions%
\begin{equation}
\left\{  \Omega\left(  p^{M}\left\Vert x-\boldsymbol{i}\right\Vert
_{p}\right)  \right\}  _{\boldsymbol{i}\in G_{M}^{N}} \label{Basis1}%
\end{equation}
are orthogonal with respect to the standard $L^{2}$ inner product, since%
\[
\Omega\left(  p^{M}\left\Vert x-\boldsymbol{i}\right\Vert _{p}\right)
\Omega\left(  p^{M}\left\Vert x-\boldsymbol{j}\right\Vert _{p}\right)
=0\text{, for }\boldsymbol{i},\boldsymbol{j}\in G_{M}^{N}\text{,
}\boldsymbol{i}\neq\boldsymbol{j}\text{ and for any }x\in B_{M}^{N}.
\]
We denote by $\mathcal{D}^{M}\left(
\mathbb{Q}
_{p}^{N}\right)  :=\mathcal{D}^{M}$ the $\mathbb{R}$-vector space spanned by
(\ref{Basis1}). We set
\[
\text{ }\mathcal{X}_{M}:=\left(  \mathcal{D}^{M},\left\Vert \cdot\right\Vert
_{\infty}\right)  \text{ for }M\geq1\text{.}
\]
Notice that $\mathcal{X}_{M}$ is isomorphic as a Banach space to $\left(
\mathbb{R}^{\#G_{M}^{N}},\left\Vert \cdot\right\Vert _{\mathbb{R}}\right)  $,
where
\[
\left\Vert \left(  t_{1},\ldots,t_{\#G_{M}^{N}}\right)  \right\Vert
_{\mathbb{R}}=\max_{1\leq j\leq\#G_{M}^{N}}\left\vert t_{j}\right\vert .
\]

\subsection{Tree-like structures and $p$-adic numbers}

Take $N=1$ and fix $M\in\mathbb{N\smallsetminus}\left\{  0\right\}  $, then
$G_{M}^{1}:=G_{M}=p^{-M}\mathbb{Z}_{p}/p^{M}\mathbb{Z}_{p}$ is an\ additive
group consisting of elements of the form%
\begin{equation}
\boldsymbol{i}=\boldsymbol{i}_{-M}p^{-M}+\boldsymbol{i}_{-M+1}p^{-M+1}%
+\cdots+\boldsymbol{i}_{0}+\cdots+\boldsymbol{i}_{M-1}p^{M-1}, \label{seq_I}%
\end{equation}
where the $\boldsymbol{i}_{j}$s belong to $\left\{  0,1,\ldots,p-1\right\}  $.
Furthermore, the restriction of $\left\vert \cdot\right\vert _{p}$ to $G_{M}$
induces an absolute value such that $\left\vert G_{M}\right\vert _{p}=\left\{
0,p^{-\left(  M-1\right)  },\cdots,p^{-1},1,\cdots,p^{M}\right\}  $. We endow
$G_{M}$ with the metric induced by $\left\vert \cdot\right\vert _{p}$, and
thus $G_{M}$ becomes a finite ultrametric space. In addition, $G_{M}$ can be
identified with the set of branches (vertices at the top level) of a rooted
tree with $2M+1$ levels and $p^{2M}$ branches. Any element $\boldsymbol{i}\in
G_{M}$ can be uniquely written as $p^{-M}\widetilde{\boldsymbol{i}}$, where%
\[
\widetilde{\boldsymbol{i}}=\widetilde{\boldsymbol{i}}_{0}+\widetilde
{\boldsymbol{i}}_{1}p+\cdots+\widetilde{\boldsymbol{i}}_{2M-1}p^{2M-1}%
\in\mathbb{Z}_{p}/p^{2M}\mathbb{Z}_{p},
\]
with the $\widetilde{\boldsymbol{i}}_{j}$s belonging to $\left\{
0,1,\ldots,p-1\right\}  $. The elements of the $\mathbb{Z}_{p}/p^{2M}%
\mathbb{Z}_{p}$ are in bijection with the vertices at the top level of the
above mentioned rooted tree. By definition the root of the tree is the only
vertex at level $0$. There are exactly $p$ vertices at level $1$, which
correspond with the possible values of the digit $\widetilde{\boldsymbol{i}%
}_{0}$ in the $p$-adic expansion of $\widetilde{\boldsymbol{i}}$. Each of
these vertices is connected to the root by a non-directed edge. At level
$\ell$, with $1\leq\ell\leq2M$, there are exactly $p^{\ell}$ vertices, \ each
vertex corresponds to a truncated expansion of $\widetilde{\boldsymbol{i}}$ of
the form $\widetilde{\boldsymbol{i}}_{0}+\cdots+\widetilde{\boldsymbol{i}%
}_{\ell-1}p^{\ell-1}$. The vertex corresponding to $\widetilde{\boldsymbol{i}%
}_{0}+\cdots+\widetilde{\boldsymbol{i}}_{\ell-1}p^{\ell-1}$ is connected to a
vertex $\widetilde{\boldsymbol{i}}_{0}^{\prime}+\cdots+\widetilde
{\boldsymbol{i}}_{\ell-2}^{\prime}p^{\ell-2}$ at the level $\ell-1$ if and
only if $\left(  \widetilde{\boldsymbol{i}}_{0}+\cdots+\widetilde
{\boldsymbol{i}}_{\ell-1}p^{\ell-1}\right)  -\left(  \widetilde{\boldsymbol{i}%
}_{0}^{\prime}+\cdots+\widetilde{\boldsymbol{i}}_{\ell-2}^{\prime}p^{\ell
-2}\right)  $ is divisible by $p^{\ell-1}$.

In conclusion, $\mathbb{Z}_{p}/p^{2M}\mathbb{Z}_{p}$ is a rooted tree, and
$\mathbb{Z}_{p}$ is an infinite rooted tree. Now, the $1$-dimensional unit
sphere $\mathbb{Z}_{p}^{\times}$ is the disjoint union of sets of the form
$j+p\mathbb{Z}_{p}$, for $j\in\left\{  1,\ldots,p-1\right\}  $. Each set of
the form $j+p\mathbb{Z}_{p}$ is an infinite rooted tree. Then, $\mathbb{Z}%
_{p}^{\times}$ is a forest formed by the disjoint union of $p-1$ infinite
rooted trees. On the other hand, $\mathbb{Q}_{p}^{\times}=\mathbb{Q}%
_{p}\smallsetminus\left\{  0\right\}  $ is a countable disjoint union of
scaled versions of the forest $\mathbb{Z}_{p}^{\times}$, more precisely,
$\mathbb{Q}_{p}^{\times}=%
{\textstyle\bigsqcup\nolimits_{k=-\infty}^{k=+\infty}}
p^{k}\mathbb{Z}_{p}^{\times}$. The field of $p$-adic numbers has a fractal
structure, see e.g. \cite{Alberio et al}, \cite{V-V-Z}.%

\begin{figure}
[h]
\begin{center}
\includegraphics[
height=1.8585in,
width=4.7816in
]%
{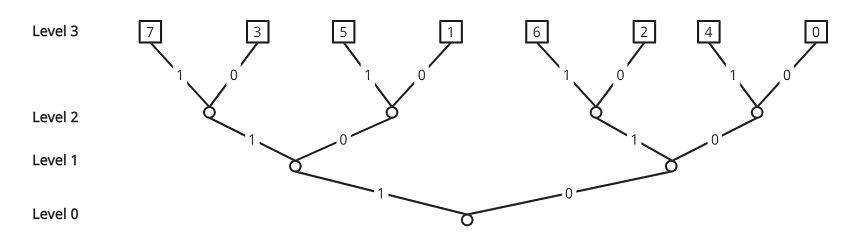}%
\caption{The rooted tree associated with the group $\mathbb{Z}_{2}%
/2^{3}\mathbb{Z}_{2}$. We identify the elements of $\mathbb{Z}_{2}%
/2^{3}\mathbb{Z}_{2}$ with the set of integers $\left\{  0,\ldots,7\right\}  $
\ with binary representation $\boldsymbol{i}=\boldsymbol{i}_{0}+\boldsymbol{i}%
_{1}2+\boldsymbol{i}_{3}2^{2},\;\;\;\boldsymbol{i}_{0},\boldsymbol{i}%
_{1},\boldsymbol{i}_{2}\in\{0,1\}$. Two leaves $\boldsymbol{i},\boldsymbol{j}%
\in\mathbb{Z}_{2}/2^{3}\mathbb{Z}_{2}$ have a common ancestor at level $2$ if
and only if $\boldsymbol{i}\equiv\boldsymbol{j}$ $\operatorname{mod}$ $2^{2}$,
i.e., $\boldsymbol{i}=\boldsymbol{a}_{0}+\boldsymbol{a}_{1}2+\boldsymbol{i}%
_{2}2^{2}$ and $\boldsymbol{j}=\boldsymbol{a}_{0}+\boldsymbol{a}%
_{1}2+\boldsymbol{j}_{2}2^{2}$ with $\boldsymbol{i}_{2},\boldsymbol{j}_{2}%
\in\{0,1\}$. Now, for $\boldsymbol{i},\boldsymbol{j\in}\mathbb{Z}_{2}%
/2^{3}\mathbb{Z}_{2}$ have a common ancestor at level $1$ if and only if
$\boldsymbol{i}\equiv\boldsymbol{j}$ $\operatorname{mod}$ $2$. Notice that
that the $p$-adic distance satisfies $-\log_{2}\left\vert \boldsymbol{i}%
-\boldsymbol{j}\right\vert _{2}=-\left(  \text{level of the first common
ancestor of }\boldsymbol{i}\text{, }\boldsymbol{j}\right)  $.}%
\end{center}
\end{figure}

\section{$p$-Adic CNNs: basic definitions}

We say that a \ function $f:\mathbb{R}\rightarrow\mathbb{R}$ is called a
Lipschitz function if there exists a real constant $L(f)>0$ such that, for all
$x,y\in\mathbb{R}$, $|f(x)-f(y)|\leq L(f)|x-y|$. A relevant example is
\[
f(x)=\frac{1}{2}\left(  |x+1|-|x-1|\right)  .
\]

\subsection{$p$-Adic discrete CNNs}

By considering $G_{M}^{N}$ as a subset of $\mathbb{Q}_{p}^{N}$, $\left(
G_{M}^{N},\left\Vert \cdot\right\Vert _{p}\right)  $ becomes a finite
ultrametric space.%

\begin{figure}
[h]
\begin{center}
\includegraphics[
height=2.2693in,
width=2.5642in
]%
{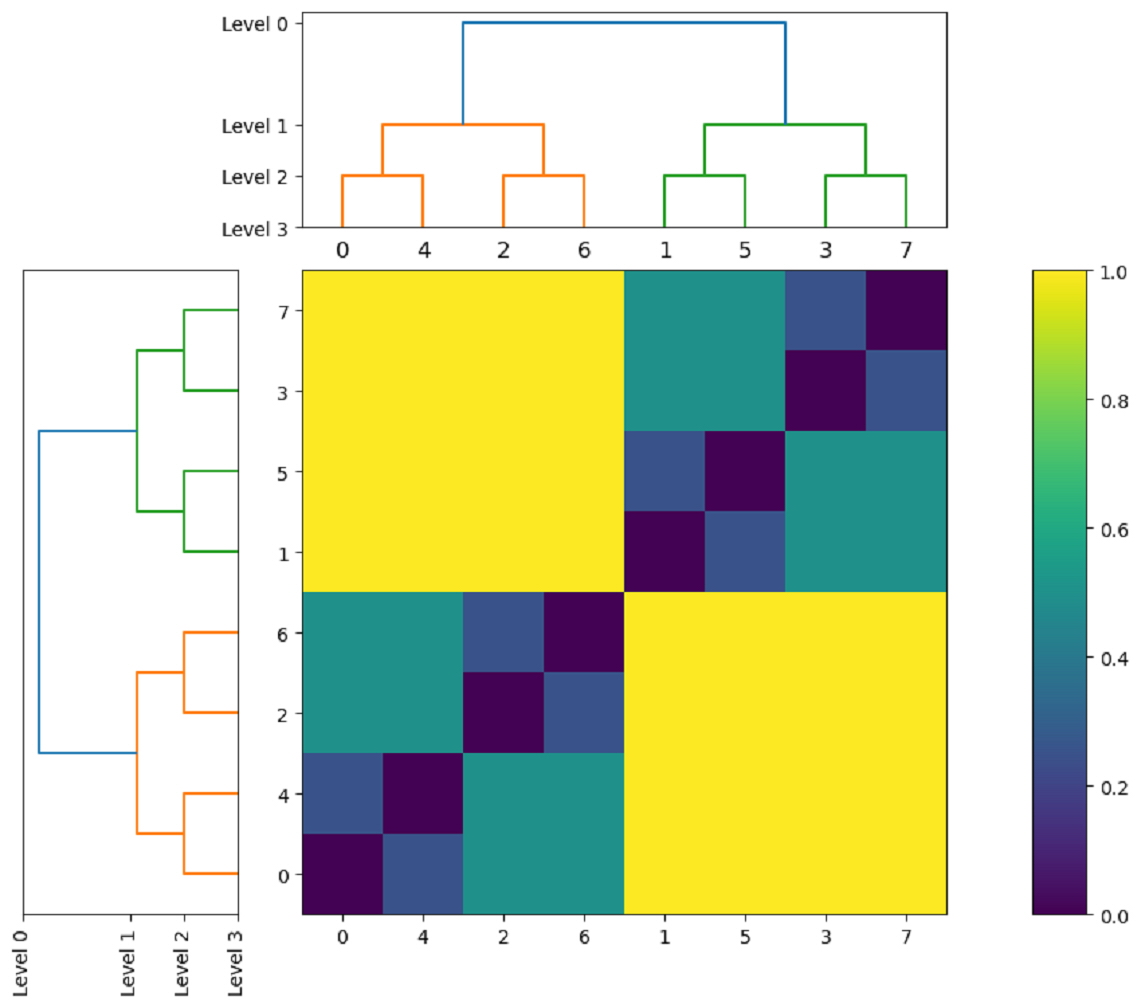}%
\caption{The heat map associated with the $p$-adic distance function on
$\mathbb{Z}_{2}/2^{3}\mathbb{Z}_{2}$.}%
\end{center}
\end{figure}

\begin{definition}
An element $\boldsymbol{i}$ of $G_{M}^{N}$ is called a cell. A $p$-adic
discrete CNN is a dynamical system CNN$_{d}(\mathbb{A},\mathbb{B},U,Z)$ on
$G_{M}^{N}$. The state $X_{\boldsymbol{i}}(t)\in\mathbb{R}$ of cell
$\boldsymbol{i}$ is described by the following differential equations:

\begin{itemize}
\item[(i)] state equation:
\[
\frac{dX({\boldsymbol{i}},t)}{dt}=-X({\boldsymbol{i}},t)+\sum_{\boldsymbol{j}\in
G_{M}^{N}}\mathbb{A}(\boldsymbol{i},\boldsymbol{j})Y({\boldsymbol{j}},t)+\sum_{\boldsymbol{j}\in G_{M}^{N}}\mathbb{B}(\boldsymbol{i},\boldsymbol{j}%
)U(\boldsymbol{j})+Z(\boldsymbol{i})\text{, }\boldsymbol{i}\in G_{M}%
^{N}\text{,\ }%
\]

\item[(ii)] output equation:
\[
Y({\boldsymbol{j}},t)=f(X({\boldsymbol{j}},t)),
\]
where $Y({\boldsymbol{i}},t)\in\mathbb{R}$ is the output of cell
$\boldsymbol{i}$ at the time $t$, $f:\mathbb{R}\rightarrow\mathbb{R}$ is a
bounded Lipschitz function satisfying $f(0)=0$. The function $U(\boldsymbol{i}%
)\in\mathbb{R}$ is the input of the cell $\boldsymbol{i}$, $Z(\boldsymbol{i}%
)\in\mathbb{R}$ is the threshold of cell $\boldsymbol{i}$, and $\mathbb{A}%
,\mathbb{B}:G_{M}^{N}\times G_{M}^{N}\rightarrow\mathbb{R}$ are the feedback
operator and feedforward operator, respectively.
\end{itemize}
\end{definition}

Not all the cells of $G_{M}^{N}$ are active. A cell $\boldsymbol{i}$ is
connected with cell $\boldsymbol{j}$ if $\mathbb{A}(\boldsymbol{i}%
,\boldsymbol{j})\neq0$\ or $\mathbb{B}(\boldsymbol{i},\boldsymbol{j})\neq0$
for some $\boldsymbol{j}\in G_{M}^{N}$. Then, a $p$-adic discrete CNN is a
dynamical system on
\[
C_{N,M}:=\left\{  \boldsymbol{i}\in G_{M}^{N};\mathbb{A}(\boldsymbol{i}%
,\boldsymbol{j})\neq0\ \text{or }\mathbb{B}(\boldsymbol{i},\boldsymbol{j}%
)\neq0\text{ for some }\boldsymbol{j}\in G_{M}^{N}\right\}  .
\]
The topology of a\ $p$-adic discrete $\text{CNN depends on the func\-tions
}\mathbb{A}$, $\mathbb{B}:$ $G_{M}^{N}\times G_{M}^{N}\rightarrow\mathbb{R}$.
For general matrices $\mathbb{A}$, $\mathbb{B}$, it is difficult to give a
graph-type description of the topology of the network. Our $p$-adic CNNs
contain as a particular case the CNNs of Chua and Yang, see e.g.
\cite{Chua-Tamas}, \cite{Slavova}. In this article we focus on $p$-adic CNNs
\ satisfying%
\begin{equation}
\mathbb{A}(\boldsymbol{i},\boldsymbol{j})=\mathbb{A}(\left\Vert \boldsymbol{i}%
-\boldsymbol{j}\right\Vert _{p})\text{, }\mathbb{B}(\boldsymbol{i}%
,\boldsymbol{j})=\mathbb{B}(\left\Vert \boldsymbol{i}-\boldsymbol{j}%
\right\Vert _{p}), \label{Hyp_1}%
\end{equation}
which are discrete CNNs having the space-invariant property. The fact that
$\mathbb{A}$ and $\mathbb{B}$ are radial functions of $\Vert\cdot\Vert_{p}$
implies that the cells are organized in a tree like-structure with many layers.%

\begin{figure}
[h]
\begin{center}
\includegraphics[
height=2.2771in,
width=4.9632in
]%
{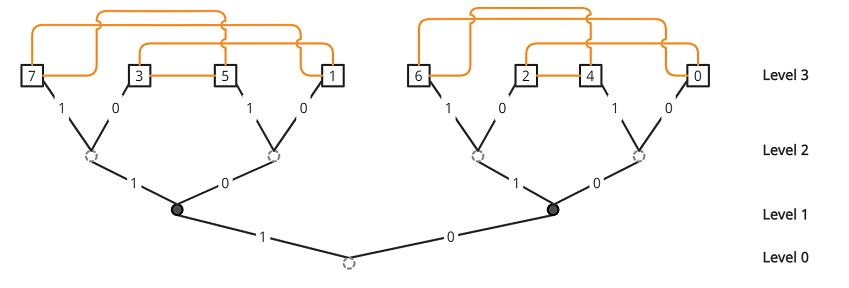}%
\caption{A $1$-dimensional discrete $2$-adic CNN with $8$ cells:
$C_{1,3}=\{0,1,2,3,4,5,7\}\subset\mathbb{Z}_{2}/2^{3}\mathbb{Z}_{2}%
\subset2^{-3}\mathbb{Z}_{2}/2^{3}\mathbb{Z}_{2}$. We set $\mathbb{B}=0$ and
$\mathbb{A}(\boldsymbol{i},\boldsymbol{j})=\left[  a_{\boldsymbol{i}%
,\boldsymbol{j}}\right]  $, with $a_{\boldsymbol{i},\boldsymbol{j}}\neq0$ if
$|\boldsymbol{i}-\boldsymbol{j}|_{2}=1/2$ and $\boldsymbol{i}$,
$\boldsymbol{j}\in C_{1,3}$; $a_{\boldsymbol{i},\boldsymbol{j}}=0$ otherwise.}%
\end{center}
\end{figure}
%

\begin{figure}
[h]
\begin{center}
\includegraphics[
height=2.2321in,
width=2.8971in
]%
{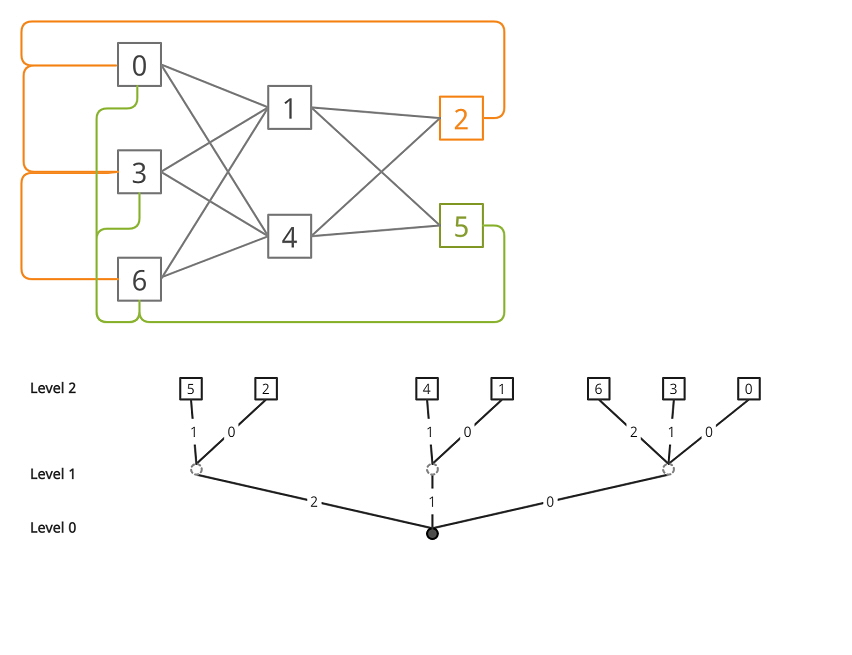}%
\caption{A $1$-dimensional $3$-adic CNN with $7$ cells, $C_{1,2}%
=\{0,1,2,3,4,5,6\}\subset\mathbb{Z}_{3}/3^{2}\mathbb{Z}_{3}\subset
3^{-2}\mathbb{Z}_{3}/3^{2}\mathbb{Z}_{3}$. We set $\mathbb{B}=0$ and
$\mathbb{A}(\boldsymbol{i},\boldsymbol{j})=\left[  a_{\boldsymbol{i}%
,\boldsymbol{j}}\right]  $, with $a_{\boldsymbol{i},\boldsymbol{j}}\neq0$ if
$|\boldsymbol{i}-\boldsymbol{j}|_{3}=1$ and $\boldsymbol{i},\boldsymbol{j}\in
C_{1,2}$; $a_{\boldsymbol{i},\boldsymbol{j}}=0$ otherwise.}%
\end{center}
\end{figure}

\subsection{$p$-Adic continuous CNNs}

\begin{definition}
Given $A(x,y)$, $B(x,y)\in L^{1}(\mathbb{Q}_{p}^{N}\times\mathbb{Q}_{p}^{N})$,
and $U$, $Z\in\mathcal{X}_{\infty}$, a $p$-adic continuous CNN, denoted as
CNN$(A,B,U,Z)$, is the dynamical system given by the following differential
equations: (i) state equation:%
\begin{equation}
\frac{\partial X(x,t)}{\partial t}=-X(x,t)+%
{\displaystyle\int\limits_{\mathbb{Q}_{p}^{N}}}
A(x,y)Y(y,t)d^{N}y+%
{\displaystyle\int\limits_{\mathbb{Q}_{p}^{N}}}
B(x,y)U(y)d^{N}y+Z(x), \label{Continuous_CNN}%
\end{equation}
where $x\in\mathbb{Q}_{p}^{N}$, $t\geq0$, and (ii) output equation:
$Y(x,t)=f(X(x,t))$. We say that $X(x,t)\in\mathbb{R}$ is the \textit{state of
cell} $x$ at the time $t$, $Y(x,t)\in\mathbb{R}$ is \textit{the output of
cell} $x$ at the time $t$. Function $A(x,y)$ is \textit{the kernel of the
feedback operator, while \ function } $B(x,y)$ is t\textit{he kernel of the
feedforward operator}. Function $U$ is \textit{the input of the $\text{CNN}$},
while function $Z$ is \textit{the threshold of the $\text{CNN.}$}
\end{definition}

We focus mainly in continuous CNNs having the space invariant property, i.e.
$A(x,y)=A(\Vert x-y\Vert_{p})$ and $B(x,y)=B(\Vert x-y\Vert_{p})$ for some
$A,B\in L^{1}$, however our results are valid for general $p$-adic continuous
$\text{CNN}$s. Along this article the function $f(x)=\frac{1}{2}\left(
|x+1|-|x-1|\right)  $ will be fixed, for this reason it does not appear in the
list of parameters of the CNNs.

\subsection{\label{Section_Discretization}Discretization of $p$-adic
continuous CNNs}

A central result of the present work is the fact that $p$-adic continuous CNNs
are\ `continuous versions' of $p$-adic discrete CNNs. More precisely, $p$-adic
discrete CNNs are very good approximations of $p$-adic continuous CNNs for
sufficiently large $M$. We discuss here the discretization process in an
intutive way (a formal theorem will be provided later on).

Intuitively, a discretization of a $p$-adic continuous $\text{CNN}(A,B,U,Z)$
is obtained assuming that $X(\cdot,t)$, $A$, $Y(\cdot,t)$, $B$, $U$ and $Z$
belong to $\mathcal{D}^{M}$, i.e.%
\begin{gather*}
X(x,t)=%
{\displaystyle\sum\limits_{\boldsymbol{i}\in G_{M}^{N}}}
X(\boldsymbol{i},t)\Omega\left(  p^{M}\left\Vert x-\boldsymbol{i}\right\Vert
_{p}\right)  \text{, \ }Y(x,t)=%
{\displaystyle\sum\limits_{\boldsymbol{i}\in G_{M}^{N}}}
Y(\boldsymbol{i},t)\Omega\left(  p^{M}\left\Vert x-\boldsymbol{i}\right\Vert
_{p}\right)  ,\\
U(x)=%
{\displaystyle\sum\limits_{\boldsymbol{i}\in G_{M}^{N}}}
U(\boldsymbol{i})\Omega\left(  p^{M}\left\Vert x-\boldsymbol{i}\right\Vert
_{p}\right)  \text{, \ }Z(x)=%
{\displaystyle\sum\limits_{\boldsymbol{i}\in G_{M}^{N}}}
Z(\boldsymbol{i})\Omega\left(  p^{M}\left\Vert x-\boldsymbol{i}\right\Vert
_{p}\right)  ,\\
A(x,y)=%
{\displaystyle\sum\limits_{\boldsymbol{i}\in G_{M}^{N}}}
\text{ }%
{\displaystyle\sum\limits_{\boldsymbol{j}\in G_{M}^{N}}}
A(\boldsymbol{i},\boldsymbol{j})\Omega\left(  p^{M}\left\Vert x-\boldsymbol{i}%
\right\Vert _{p}\right)  \Omega\left(  p^{M}\left\Vert y-\boldsymbol{j}%
\right\Vert _{p}\right)  ,\\
B(x,y)=%
{\displaystyle\sum\limits_{\boldsymbol{i}\in G_{M}^{N}}}
\text{ }%
{\displaystyle\sum\limits_{\boldsymbol{j}\in G_{M}^{N}}}
B(\boldsymbol{i},\boldsymbol{j})\Omega\left(  p^{M}\left\Vert x-\boldsymbol{i}%
\right\Vert _{p}\right)  \Omega\left(  p^{M}\left\Vert y-\boldsymbol{j}%
\right\Vert _{p}\right)  .
\end{gather*}
Notice that if $f:\mathbb{R}\rightarrow\mathbb{R}$, then
\[
f\left(  X(x,t)\right)  =%
{\displaystyle\sum\limits_{\boldsymbol{i}\in G_{M}^{N}}}
f(X(\boldsymbol{i},t))\Omega\left(  p^{M}\left\Vert x-\boldsymbol{i}%
\right\Vert _{p}\right)  =Y(x,t)\text{.}%
\]
Now,%
\[
\frac{\partial}{\partial t}X(x,t)=%
{\displaystyle\sum\limits_{\boldsymbol{i}\in G_{M}^{N}}}
\frac{\partial}{\partial t}X(\boldsymbol{i},t)\Omega\left(  p^{M}\left\Vert
x-\boldsymbol{i}\right\Vert _{p}\right)  ,
\]
and%
\begin{align*}
&
{\displaystyle\int\limits_{\mathbb{Q}_{p}^{N}}}
A(x,y)f\left(  X(y,t)\right)  d^{N}y\\
&  =%
{\displaystyle\sum\limits_{\boldsymbol{i}\in G_{M}^{N}}}
\left\{
{\displaystyle\sum\limits_{\boldsymbol{j}\in G_{M}^{N}}}
A(\boldsymbol{i},\boldsymbol{j})f(X(\boldsymbol{j},t))%
{\displaystyle\int\limits_{\mathbb{Q}_{p}^{N}}}
\Omega\left(  p^{M}\left\Vert y-\boldsymbol{j}\right\Vert _{p}\right)
d^{N}y\right\}  \Omega\left(  p^{M}\left\Vert x-\boldsymbol{i}\right\Vert
_{p}\right) \\
&  =p^{-MN}%
{\displaystyle\sum\limits_{\boldsymbol{i}\in G_{M}^{N}}}
\left\{
{\displaystyle\sum\limits_{\boldsymbol{j}\in G_{M}^{N}}}
A(\boldsymbol{i},\boldsymbol{j})Y(\boldsymbol{j},t))\right\}  \Omega\left(
p^{M}\left\Vert x-\boldsymbol{i}\right\Vert _{p}\right)  .
\end{align*}
Similarly,%
\[%
{\displaystyle\int\limits_{\mathbb{Q}_{p}^{N}}}
B(x,y)U(y)d^{N}y=p^{-MN}%
{\displaystyle\sum\limits_{\boldsymbol{i}\in G_{M}^{N}}}
\left\{
{\displaystyle\sum\limits_{\boldsymbol{j}\in G_{M}^{N}}}
B(\boldsymbol{i},\boldsymbol{j})U(\boldsymbol{j}))\right\}  \Omega\left(
p^{M}\left\Vert x-\boldsymbol{i}\right\Vert _{p}\right)  .
\]
Therefore,%
\begin{align*}
\frac{\partial}{\partial t}X(\boldsymbol{i},t)  &  =-X(\boldsymbol{i},t)+%
{\displaystyle\sum\limits_{\boldsymbol{j}\in G_{M}^{N}}}
p^{-MN}A(\boldsymbol{i},\boldsymbol{j})Y(\boldsymbol{j},t)\\
&  +%
{\displaystyle\sum\limits_{\boldsymbol{j}\in G_{M}^{N}}}
p^{-MN}B(\boldsymbol{i},\boldsymbol{j})U(\boldsymbol{j})+Z(\boldsymbol{i}%
)\text{, for }\boldsymbol{i}\in G_{M}^{N}\text{,}%
\end{align*}
and $Y(\boldsymbol{i},t)=f(X(\boldsymbol{i},t))$, for $\boldsymbol{i}\in
G_{M}^{N}$. This is exactly a $p$-adic discrete CNN with $\mathbb{A}%
(\boldsymbol{i},\boldsymbol{j})=p^{-MN}A(\boldsymbol{i},\boldsymbol{j})$,
$\mathbb{B}(\boldsymbol{i},\boldsymbol{j})=p^{-MN}B(\boldsymbol{i}%
,\boldsymbol{j})$.

Intuitively a $p$-adic continuous CNN has infinitely many layers, each layer
corresponds to some $M$, which are organized in a hierarchical structure. For
practical purposes, a $p$-adic continuous CNN is realized as a $p$-adic
discrete CNN for $M$ sufficiently large.

\section{Stability of $p$-adic continuous CNN}

\begin{lemma}
\label{Lemma 1} Let $f$ be a Lipschitz functions with $f(0)=0$ and let $E$ be
a radial function in $L^{1}(\mathbb{Q}_{p}^{N})$. Then, the mappings
\[%
\begin{split}
&  F_{0}:g\rightarrow\int_{\mathbb{Q}_{p}^{N}}E(\Vert x-y\Vert_{p}%
)f(g(y))d^{N}y\\
&  F_{1}:g\rightarrow\int_{\mathbb{Q}_{p}^{N}}E(\Vert x-y\Vert_{p})g(y)d^{N}y
\end{split}
\]
are well defined bounded operators from $\mathcal{X}_{\infty}$ into itself.
\end{lemma}

\begin{proof}
We first notice that for all $g\in\mathcal{X}_{\infty}$, $F_{0}(g)(x)$ exists
for all $x\in\mathbb{Q}_{p}^{N}$, since
\begin{equation}
|E(\Vert y\Vert_{p})|\text{ }|f(g(x-y))|\leq L(f)\Vert g\Vert_{\infty}\text{
}\left\vert E(\Vert y\Vert_{p})\right\vert , \label{eq:ast}%
\end{equation}
where $E(\Vert y\Vert_{p})\in L^{1}(\mathbb{Q}_{p}^{N})$. To show the
continuity of $F_{0}(g)(x)$, we take a sequence $\{x_{m}\}_{m\in\mathbb{N}%
}\subset$ $\mathbb{Q}_{p}^{N}$ such that $x_{m}\rightarrow x$. By using
(\ref{eq:ast}) and the dominated convergence theorem, $\lim_{m\rightarrow
\infty}F_{0}(g)(x_{m})=F_{0}(g)(x)$. Finally, we show that $F_{0}%
(g)\in\mathcal{X}_{\infty}$. By contradiction, assume that $F_{0}%
(g)\not \in \mathcal{X}_{\infty}$. Then, there is a sequence $\{x_{m}%
\}_{m\in\mathbb{N}}\subset$ $\mathbb{Q}_{p}^{N}$ such that $\lim
_{m\rightarrow\infty}\Vert x_{m}\Vert_{p}=\infty$ and $\epsilon>0$ such that
$F_{0}(g)(x_{m})>\epsilon$ for all $m\in\mathbb{N}$. By using (\ref{eq:ast})
and the dominated convergence theorem, we have%
\begin{align*}
\epsilon &  \leq\lim_{m\rightarrow\infty}|F_{0}(g)(x_{m})|=\lim_{m\rightarrow
\infty}\left\vert \int_{\mathbb{Q}_{p}^{N}}E(\Vert y\Vert_{p})f(g(x_{m}%
-y))d^{N}y\right\vert \\
&  =\left\vert \int_{\mathbb{Q}_{p}^{N}}E(\Vert y\Vert_{p})\left\{
\lim_{m\rightarrow\infty}f(g(x_{m}-y))\right\}  d^{N}y\right\vert =0
\end{align*}
which contradicts the fact $\epsilon>0$. The same argument allow us to show
that $F_{1}(g)\in\mathcal{X}_{\infty}$ for any $g\in\mathcal{X}_{\infty}$.
\end{proof}

\begin{lemma}
\label{Lemma3} Assume $A,B\in L^{1}(\mathbb{Q}_{p}^{N})$ are radial functions
and that $U$, $Z\in\mathcal{X}_{\infty}$. For $g\in\mathcal{X}_{\infty}$, set
\[
\boldsymbol{H}(g):=%
{\displaystyle\int\limits_{\mathbb{Q}_{p}^{N}}}
A(\Vert x-y\Vert_{p})f\left(  g(y)\right)  d^{N}y+%
{\displaystyle\int\limits_{\mathbb{Q}_{p}^{N}}}
B(\Vert x-y\Vert_{p})U(y)d^{N}y+Z(x).
\]
Then $\boldsymbol{H}:\mathcal{X}_{\infty}\rightarrow\mathcal{X}_{\infty}$ is a
well-defined \ operator satisfying%
\[
\Vert\boldsymbol{H}(g)-\boldsymbol{H}(g^{\prime})\Vert_{\infty}\leq L(f)\Vert
A\Vert_{1}\Vert g-g^{\prime}\Vert_{\infty}\text{, for }g\text{, }g^{\prime}%
\in\mathcal{X}_{\infty}\text{,}
\]
where $L(f)$ is the Lipschitz constant of $f$.
\end{lemma}

\begin{proof}
By Lemma \ref{Lemma 1}, $\boldsymbol{H}:\mathcal{X}_{\infty}\rightarrow
\mathcal{X}_{\infty}$ is a well-defined \ operator. Take $g$, $g^{\prime}%
\in\mathcal{X}_{\infty}$, then {\small
\begin{multline*}
|\boldsymbol{H}(g)(x)-\boldsymbol{H}(g^{\prime})(x)|=\left\vert \text{ }%
{\displaystyle\int\limits_{\mathbb{Q}_{p}^{N}}}
A(\Vert x-y\Vert_{p})\left(  f\left(  g(y)\right)  -f\left(  g^{\prime
}(y)\right)  \right)  d^{N}y\right\vert \\
\leq%
{\displaystyle\int\limits_{\mathbb{Q}_{p}^{N}}}
|A(\Vert x-y\Vert_{p})||f\left(  g(y)\right)  -f\left(  g^{\prime}(y)\right)
|d^{N}y\leq L(f)\Vert g-g^{\prime}\Vert_{\infty}%
{\displaystyle\int\limits_{\mathbb{Q}_{p}^{N}}}
|A(\Vert x-y\Vert_{p})|d^{N}y\\
=L(f)\Vert A\Vert_{1}\Vert g-g^{\prime}\Vert_{\infty}.
\end{multline*}
}
\end{proof}

\begin{remark}
\label{nota_lemma2}(i) Lemma \ref{Lemma 1} remains valid if we replace the
condition $E$ is radial and integrable by the condition $E(x,y)$ is a
continuous function with compact support.

\noindent(ii) Under the hypothesis of part (i), Lemma \ref{Lemma3} is valid
for operators of the form%
\[
\boldsymbol{L}g=%
{\displaystyle\int\limits_{\mathbb{Q}_{p}^{N}}}
A(x,y)f\left(  g(y)\right)  d^{N}y+%
{\displaystyle\int\limits_{\mathbb{Q}_{p}^{N}}}
B(x,y)U(y)d^{N}y+Z(x),
\]
for $g\in\mathcal{X}_{\infty}$.
\end{remark}

\begin{proposition}
\label{Prop1}Assume that $A$, $B$, $f$ satisfy hypotheses of Lemma
\ref{Lemma3} and that $U$, $Z\in\mathcal{X}_{\infty}$. Let $\tau$ be a fixed
positive real number. Then for each $X_{0}\in\mathcal{X}_{\infty}$ there
exists a unique $X\in C([0,\tau],\mathcal{X}_{\infty})$ which satisfies
\begin{equation}
X(x,t)=e^{-t}X_{0}\left(  x\right)  +\int_{0}^{t}e^{-(t-s)}\boldsymbol{H}%
(X(x,s))ds \label{eq:solution}%
\end{equation}
where
\begin{equation}
\boldsymbol{H}X(x,t)=%
{\displaystyle\int\limits_{\mathbb{Q}_{p}^{N}}}
A(\Vert x-y\Vert_{p})f(X(y,t))d^{N}y+%
{\displaystyle\int\limits_{\mathbb{Q}_{p}^{N}}}
B(\Vert x-y\Vert_{p})U(y)d^{N}y+Z(x). \label{eq:Op_H}%
\end{equation}
The function $X(x,t)$ is differentiable in $t$ for all $x$, and it is a
solution of equation (\ref{Continuous_CNN}) with initial datum $X_{0}$.
\end{proposition}

\begin{proof}
The result follows from Lemma \ref{Lemma3}, by using standard techniques in
PDEs, see e.g. \cite[Theorem 5.1.2]{Milan}. To make the treatment
comprehensive to a general audience, we provide some details here. First,
define
\[
\boldsymbol{T}(Y)=X_{0}e^{-t}+\int_{0}^{t}e^{-(t-s)}\boldsymbol{H}(Y(x,s))ds,
\]
and $\mathcal{Y}=C([0,\tau],\mathcal{X}_{\infty})$ which is a Banach space
with the norm $\Vert\cdot\Vert_{\infty}$. By Lemma \ref{Lemma3},
$\boldsymbol{T}:\mathcal{Y}\rightarrow\mathcal{Y}$. \ If $Y$, $Y_{1}%
\in\mathcal{Y}$, then
\begin{gather*}
\Vert\boldsymbol{T}(Y)(t)-\boldsymbol{T}(Y_{1})(t)\Vert_{\infty}=\left\Vert
\int_{0}^{t}e^{-(t-s)}\left\{  \boldsymbol{H}(Y)(s)-\boldsymbol{H}%
(Y_{1})(s)\right\}  ds\right\Vert _{\infty}\\
\leq\int_{0}^{t}e^{-(t-s)}\Vert\boldsymbol{H}(Y)(s)-\boldsymbol{H}%
(Y_{1})(s)\Vert_{_{\infty}}\text{ }ds\leq L(f)\Vert A\Vert_{1}\int_{0}%
^{t}\Vert Y-Y_{1}\Vert_{_{\infty}}ds.
\end{gather*}
And hence,
\[
\Vert\boldsymbol{T}^{M}(Y)(t)-\boldsymbol{T}^{M}(Y_{1})(t)\Vert_{\infty}%
\leq\frac{\tau^{M}L(f)^{M}\Vert A\Vert_{1}^{M}}{M!}\Vert Y-Y_{1}\Vert_{\infty
},
\]
for $M\geq1$. By the contraction mapping theorem, there is a unique unique
$X\in\mathcal{Y}$ which $\boldsymbol{T}(X)=X$. Moreover, since the right-hand
side of (\ref{eq:solution}) is differentiable in $t$, $X$ is a solution of
(\ref{Continuous_CNN}) with initial condition $X_{0}$.
\end{proof}

\begin{remark}
\label{Nota1}The contraction mapping theorem provides an iterative formula for
$X(x,t)$. Set $X_{1}(x,t)=X_{0}\left(  x\right)  $ and
\[
X_{L+1}(x,t)=e^{-t}X_{0}\left(  x\right)  +\int_{0}^{t}e^{-(t-s)}%
H(X_{L}(x,s))ds\text{, for }L=1,2,\ldots,
\]
then $\lim_{L\rightarrow\infty}\Vert X_{L}\left(  \cdot,t\right)  -X\left(
\cdot,t\right)  \Vert_{\infty}$ $=0$ for each $t\leq\tau$, see e.g.
\cite[Theorem 5.2.2]{Milan}.
\end{remark}

\begin{theorem}
\label{Theorem0}Assume $A$, $B\in L^{1}(p^{-M_{0}}\mathbb{Z}_{p}^{N})$ are
radial functions, for some $M_{0}\in\mathbb{N}$, and that $U$, $Z$, $X_{0}%
\in\mathcal{X}_{M_{0}}$. We also assume that $f$ is a Lipschitz functions with
$f(0)=0$. Then there is a unique $X\in C([0,\tau],\mathcal{X}_{M_{0}})\cap
C^{1}([0,\tau],\mathcal{X}_{M_{0}})$ satisfying (\ref{eq:solution}), which is
a solution of equation (\ref{Continuous_CNN}) with initial datum $X_{0}$.
\end{theorem}

\begin{remark}
This theorem remains valid if $A(x,y)$, $B(x,y)$ are continuous functions with
compact support, see Remark \ref{nota_lemma2}.
\end{remark}

\begin{proof}
Since $\mathcal{X}_{M_{0}}$ is a subspace of $\mathcal{X}_{\infty}$, by
applying Proposition \ref{Prop1}, there exists a unique $X\in C([0,\tau
],\mathcal{X}_{\infty})\cap C^{1}([0,\tau],\mathcal{X}_{\infty})$ that
satisfies all the announced properties. By Remark \ref{Nota1}, $\lim
_{L\rightarrow\infty}\Vert X_{L}\left(  \cdot,t\right)  -X\left(
\cdot,t\right)  \Vert_{\infty}$ $=0$, where
\[
X_{L+1}(x,t)=e^{-t}X_{0}\left(  x\right)  +\int_{0}^{t}e^{-(t-s)}%
H(X_{L}(x,s))ds\text{, for }L=1,2,\ldots.
\]
By induction on $L$, if $X_{L}(\cdot,s)\in\mathcal{X}_{M_{0}}$, i.e. if
\begin{align*}
X_{L}(x,s)  &  =\sum_{\boldsymbol{i}\in G_{M_{0}}^{N}}X_{L}(\boldsymbol{i}%
,s)\Omega\left(  p^{M_{0}}\left\Vert x-\boldsymbol{i}\right\Vert _{p}\right)
,\\
f(X_{L}(x,s))  &  =\sum_{\boldsymbol{i}\in G_{M_{0}}^{N}}Y_{L}(\boldsymbol{i}%
,s)\Omega\left(  p^{M_{0}}\left\Vert x-\boldsymbol{i}\right\Vert _{p}\right)
\end{align*}
by using that
\begin{multline*}
\int_{0}^{t}e^{-(t-s)}H(X_{L}(x,s))ds\\
=%
{\displaystyle\sum\limits_{\boldsymbol{i}\in G_{M_{0}}^{N}}}
\left(  \int_{0}^{t}e^{-(t-s)}Y_{L}(\boldsymbol{i},s)ds\right)  \left(
{\displaystyle\int\limits_{\mathbb{Q}_{p}^{N}}}
A(\Vert x-y\Vert_{p})\Omega\left(  p^{M_{0}}\left\Vert y-\boldsymbol{i}%
\right\Vert _{p}\right)  d^{N}y\right) \\
+%
{\displaystyle\sum\limits_{\boldsymbol{i}\in G_{M_{0}}^{N}}}
U(\boldsymbol{i})(1-e^{-t})%
{\displaystyle\int\limits_{\mathbb{Q}_{p}^{N}}}
B(\Vert x-y\Vert_{p})\Omega\left(  p^{M_{0}}\left\Vert y-\boldsymbol{i}%
\right\Vert _{p}\right)  d^{N}y\\
+%
{\displaystyle\sum\limits_{\boldsymbol{i}\in G_{M_{0}}^{N}}}
(1-e^{-t})Z(\boldsymbol{i})\Omega\left(  p^{M_{0}}\left\Vert x-\boldsymbol{i}%
\right\Vert _{p}\right)  ,
\end{multline*}
and that for any $E\in L^{1}(p^{-M_{0}}\mathbb{Z}_{p}^{N})$ are radial
function, with the convention that the support of $E$ is the ball $p^{-M_{0}%
}\mathbb{Z}_{p}^{N}$,
\begin{multline*}%
{\displaystyle\int\limits_{\mathbb{Q}_{p}^{N}}}
E(\Vert x-y\Vert_{p})\Omega\left(  p^{M_{0}}\left\Vert y-\boldsymbol{i}%
\right\Vert _{p}\right)  d^{N}y=%
{\displaystyle\int\limits_{\boldsymbol{i}+p^{M_{0}}\mathbb{Z}_{p}^{N}}}
E(\Vert x-y\Vert_{p})d^{N}y\\
=\left\{
\begin{array}
[c]{lll}%
0 & \text{if} & x\notin p^{-M_{0}}\mathbb{Z}_{p}^{N}\\%
{\displaystyle\int\limits_{p^{M_{0}}\mathbb{Z}_{p}^{N}}}
E(\Vert z\Vert_{p})d^{N}z & \text{if} & x\in\boldsymbol{i}+p^{M_{0}}%
\mathbb{Z}_{p}^{N}\\
p^{-M_{0}N}E(\Vert\boldsymbol{i}-\boldsymbol{j}\Vert_{p}) & \text{if} &
x\in\boldsymbol{j}+p^{M_{0}}\mathbb{Z}_{p}^{N}\text{, }\boldsymbol{i}%
\neq\boldsymbol{j},
\end{array}
\right.
\end{multline*}
we conclude that%
\begin{gather}
X_{L+1}(x,t)=e^{-t}X_{0}\left(  x\right)  +\label{Eq_Aproximation}\\%
{\displaystyle\sum\limits_{\boldsymbol{j}\in G_{M_{0}}^{N}}}
\left\{
{\displaystyle\sum\limits_{\substack{\boldsymbol{i}\in G_{M_{0}}%
^{N}\\\boldsymbol{i}\neq\boldsymbol{j}}}}
a\left(  \boldsymbol{i},t\right)  p^{-M_{0}N}A(\Vert\boldsymbol{i}%
-\boldsymbol{j}\Vert_{p})\right\}  \Omega\left(  p^{M_{0}}\left\Vert
y-\boldsymbol{j}\right\Vert _{p}\right)  +\nonumber\\%
{\displaystyle\sum\limits_{\boldsymbol{j}\in G_{M_{0}}^{N}}}
a\left(  \boldsymbol{j},t\right)  \left(
{\displaystyle\int\limits_{p^{M_{0}}\mathbb{Z}_{p}^{N}}}
A(\Vert z\Vert_{p})d^{N}z\right)  \Omega\left(  p^{M_{0}}\left\Vert
y-\boldsymbol{j}\right\Vert _{p}\right)  +\nonumber\\%
{\displaystyle\sum\limits_{\boldsymbol{i}\in G_{M_{0}}^{N}}}
\left\{
{\displaystyle\sum\limits_{\substack{\boldsymbol{i}\in G_{M_{0}}%
^{N}\\\boldsymbol{i}\neq\boldsymbol{j}}}}
U(\boldsymbol{i})(1-e^{-t})B(\Vert\boldsymbol{i}-\boldsymbol{j}\Vert
_{p})\right\}  \Omega\left(  p^{M_{0}}\left\Vert y-\boldsymbol{j}\right\Vert
_{p}\right)  +\nonumber\\%
{\displaystyle\sum\limits_{\boldsymbol{i}\in G_{M_{0}}^{N}}}
U\left(  \boldsymbol{j}\right)  (1-e^{-t})\left(
{\displaystyle\int\limits_{p^{M_{0}}\mathbb{Z}_{p}^{N}}}
B(\Vert z\Vert_{p})d^{N}z\right)  \Omega\left(  p^{M_{0}}\left\Vert
y-\boldsymbol{j}\right\Vert _{p}\right)  +(1-e^{-t})Z(\boldsymbol{i}%
),\nonumber
\end{gather}
i.e. $X_{L+1}(\cdot,s)\in\mathcal{X}_{M}$. And consequently, $\left\{
X_{L}(\cdot,t)\right\}  _{L\in\mathbb{N}\smallsetminus\left\{  0\right\}  }$
is a sequence in $\mathcal{X}_{M}$. Since $\mathcal{X}_{M}$ is closed in
$\mathcal{X}_{\infty}$, $X\left(  \cdot,t\right)  \in\mathcal{X}_{M}$ for any
$t\leq\tau$.
\end{proof}

\begin{remark}
By using that
\[
p^{MN}%
{\displaystyle\int\limits_{p^{M}\mathbb{Z}_{p}^{N}}}
A(\Vert z\Vert_{p})d^{N}z\rightarrow A(0)\text{, }p^{MN}%
{\displaystyle\int\limits_{p^{M}\mathbb{Z}_{p}^{N}}}
B(\Vert z\Vert_{p})d^{N}z\rightarrow B(0)
\]
as $M\rightarrow\infty$, see e.g. \cite[Theorem 1.14]{Taibleson},
(\ref{Eq_Aproximation}) provides an explicit approximation of the continuous
CNN described in Theorem \ref{Theorem0}.
\end{remark}

\begin{lemma}
\label{Cor:loc-const} Let $\tau$ be a fixed positive real number, let $X(x,t)
$ be the solution given in Proposition \ref{Prop1}, with $X(x,0)=$ $X_{0}$.
Then, for all $x,y\in\mathbb{Q}_{p}^{N}$ and $t\in\left(  0,\tau\right)  $,
\[
|X(x,t)-X(y,t)|\leq|X_{0}(x)-X_{0}(y)|e^{\Vert A\Vert_{1}L(f)t}.
\]
Moreover, if $X_{0}$ is a locally-constant function, i.e. $X_{0}(x)=X_{0}(y)$
for $y\in B_{l}(x)$, with $l=l(x)\in\mathbb{Z}$, for any $x\in\mathbb{Q}%
_{p}^{N}$, then $X(\cdot,t)$ is a locally-constant function and
$X(x,t)=X(y,t)$ for $y\in B_{l}(x)$ for any $x\in\mathbb{Q}_{p}^{N}$.
\end{lemma}

\begin{proof}
Fix $x,y\in\mathbb{Q}_{p}^{N}$, the by Proposition \ref{Prop1} and Lemma
\ref{Lemma3}, for all $t\in(0,\tau]$
\begin{align*}
|X(x,t)-X(y,t)|  &  \leq e^{-t}|X_{0}(x)-X_{0}(y)|+\int_{0}^{t}e^{-(t-s)}%
|\boldsymbol{H}(X(x,s))-\boldsymbol{H}(X(y,s))|ds\\
&  \leq|X_{0}(x)-X_{0}(y)|+L(f)\Vert A\Vert_{1}\int_{0}^{t}|X(x,s)-X(y,s)|ds.
\end{align*}
Thus, by Gronwall theorem, see \cite[Theorem 5.1.1]{Milan},
\[
|X(x,t)-X(y,t)|\leq|X_{0}(x)-X_{0}(y)|e^{L(f)\Vert A\Vert_{1}t}
\]
for all $t\in\left(  0,\tau\right)  $.
\end{proof}

\begin{definition}
\label{def:stationary state} A function $X_{stat}(x):=X_{stat}(x;A,B,U,Z)\in
\mathcal{X}_{\infty}$ is called a stationary state of a $p$-adic continuous
$\text{CNN}(A,B,U,Z)$, if%
\begin{equation}
X_{stat}(x)=%
{\displaystyle\int\limits_{\mathbb{Q}_{p}^{N}}}
A(\Vert x-y\Vert_{p})Y(y)d^{N}y+%
{\displaystyle\int\limits_{\mathbb{Q}_{p}^{N}}}
B(\Vert x-y\Vert_{p})U(y)d^{N}y+Z(x),\nonumber
\end{equation}
where $Y(x)=f(X_{stat}(x))$ and $x\in\mathbb{Q}_{p}^{N}$.
\end{definition}

\begin{remark}
\label{Lemma4} If a $p$-adic continuous $\text{CNN}(A,B,U,Z)$ satisfies that
$\Vert A\Vert_{1}L(f)<1$, then the $\text{CNN}(A,B,U,Z)$ has a unique
stationary state. This follows by the fact that, under this condition,
$\boldsymbol{H}(X)$ becomes a contraction map in $\mathcal{X}_{\infty}$, cf.
Lemmas \ref{Lemma 1}, \ref{Lemma3}.
\end{remark}

\begin{theorem}
\label{Theorem1}All the states $X(x,t)$ of a $p$-adic continuous
$\text{CNN}(A,B,U,Z)$ are bounded for all time $t\geq0$. More precisely, if
\[
X_{\max}:=\Vert X_{0}\Vert_{\infty}+\Vert f\Vert_{\infty}\Vert A\Vert
_{1}+\Vert U\Vert_{\infty}\Vert B\Vert_{1}+\Vert Z\Vert_{\infty},
\]
then
\begin{equation}
|X(x,t)|\leq X_{\max}\text{ for all }t\geq0\text{ and for all }x\in
\mathbb{Q}_{p}^{N}. \label{No_Blow_up}%
\end{equation}

In addition%
\[
X_{-}\left(  x\right)  :=\lim\inf_{t\rightarrow\infty}X(x,t)\leq
X(x,t)\leq\lim\sup_{t\rightarrow\infty}X(x,t)=:X_{+}\left(  x\right)  ,
\]
for $x\in\mathbb{Q}_{p}^{N}$. If $X_{-}\left(  x\right)  =X_{+}\left(
x\right)  :=X^{\ast}(x)$,then $X^{\ast}(x)$ is a stationary solution of the
CNN$(A,B,U,Z)$ and
\begin{equation}
X^{\ast}(x)\geq-\left\Vert f\right\Vert _{\infty}\Vert A\Vert_{1}-\Vert
U\Vert_{\infty}\Vert B\Vert_{1}+Z(x). \label{Stat_Solution_2}%
\end{equation}

\end{theorem}

\begin{remark}
Condition (\ref{No_Blow_up}) implies that $X(x,t)$ does not blow-up at finite
time. The existence of a stationary state $X^{\ast}(x)$ means that the state
of each cell of a $p$-adic continuous CNN most settle at stable equilibrium
point after the transient has decayed to zero.
\end{remark}

\begin{proof}
By Proposition \ref{Prop1}, see (\ref{eq:solution})-(\ref{eq:Op_H}), by using
that $\left\vert Y(y,t)\right\vert =\left\vert f\left(  X\left(  x,t\right)
\right)  \right\vert \leq\left\Vert f\right\Vert _{\infty}$, we have%
\begin{align*}
\left\vert \boldsymbol{H}(X\left(  x,t\right)  )\right\vert  &  \leq%
{\displaystyle\int\limits_{\mathbb{Q}_{p}^{N}}}
|A(\Vert x-y\Vert_{p})||Y(y,t)|d^{N}y+%
{\displaystyle\int\limits_{\mathbb{Q}_{p}^{N}}}
|B(\Vert x-y\Vert_{p})||U(y)|d^{N}y+|Z(x)|\\
&  \leq\Vert f\Vert_{\infty}\Vert A\Vert_{1}+\Vert B\Vert_{1}\Vert
U\Vert_{\infty}+\Vert Z\Vert_{\infty}.
\end{align*}
Therefore%
\begin{align*}
\left\Vert X\left(  x,t\right)  \right\Vert _{\infty}  &  \leq e^{-t}\Vert
X_{0}\Vert_{\infty}+\int_{0}^{t}e^{-(t-s)}\left\Vert \boldsymbol{H}\left(
X\left(  x,s\right)  \right)  \right\Vert _{\infty}ds\\
&  \leq\Vert X_{0}\Vert_{\infty}+\Vert f\Vert_{\infty}\Vert A\Vert_{1}+\Vert
B\Vert_{1}\Vert U\Vert_{\infty}+\Vert Z\Vert_{\infty}.
\end{align*}

This bound is valid for any $t\in\left[  0,\tau\right]  $, but $\tau$\ is an
arbitrary, the bound is valid for any $t\geq0$.

The bound (\ref{No_Blow_up}) implies existence of the functions:%
\begin{align*}
X_{+}\left(  x\right)   &  =\lim\sup_{t\rightarrow\infty}X(x,t)=\lim
_{M\rightarrow\infty}\sup\left\{  X(x,t);t>M\right\}  ,\\
X_{-}\left(  x\right)   &  =\lim\inf_{t\rightarrow\infty}X(x,t)=\lim
_{M\rightarrow\infty}\inf\left\{  X(x,t);t>M\right\}  .
\end{align*}
Now assume that $\lim_{t\rightarrow\infty}X(x,t)=X^{\ast}(x)$ exists. By using
that
\begin{multline*}
\int_{0}^{t}e^{-(t-s)}\boldsymbol{H}(X\left(  x,s\right)  )ds=\int_{0}%
^{t}e^{-u}\boldsymbol{H}(X\left(  x,t-u\right)  )du\\
=\int_{0}^{\infty}1_{\left[  0,t\right]  }\left(  u\right)  e^{-u}%
\boldsymbol{H}(X(x,t-u))du,
\end{multline*}
and%
\[
\left\vert 1_{\left[  0,t\right]  }\left(  u\right)  e^{-u}\boldsymbol{H}%
(X(x,t-u))\right\vert \leq\left(  \Vert f\Vert_{\infty}\Vert A\Vert_{1}+\Vert
B\Vert_{1}\Vert U\Vert_{\infty}+\Vert Z\Vert_{\infty}\right)  e^{-u}\in
L^{1}(\mathbb{R}),
\]
and the dominated convergence and Lemma \ref{Lemma3}, it follows from
(\ref{eq:solution}) that%
\begin{align*}
\lim_{t\rightarrow\infty}X(x,t)  &  =\int_{0}^{\infty}e^{-u}\lim
_{t\rightarrow\infty}\left\{  1_{\left[  0,t\right]  }\left(  u\right)
\boldsymbol{H}(X(x,t-u))\right\}  du=\int_{0}^{\infty}e^{-u}\boldsymbol{H}%
(X^{\ast}(x,))du\\
&  =%
{\displaystyle\int\limits_{\mathbb{Q}_{p}^{N}}}
A(\Vert x-y\Vert_{p})f(X^{\ast}(x))d^{N}y+%
{\displaystyle\int\limits_{\mathbb{Q}_{p}^{N}}}
B(\Vert x-y\Vert_{p})U(y)d^{N}y+Z(x).
\end{align*}

\end{proof}

\section{Stability of $p$-adic discrete CNN and Approximation of Continuous
CNNs}

\subsection{The operators $\boldsymbol{P}_{M}$, $\boldsymbol{E}_{M}$}

We now define for $M\geq1$, $\boldsymbol{P}_{M}:\mathcal{X}_{\infty
}\rightarrow\mathcal{X}_{M}$ as
\[
\boldsymbol{P}_{M}\varphi\left(  x\right)  =\sum_{\boldsymbol{i}\in G_{M}^{N}%
}\varphi\left(  \boldsymbol{i}\right)  \Omega\left(  p^{M}\left\Vert
x-\boldsymbol{i}\right\Vert _{p}\right)  .
\]
Therefore $\boldsymbol{P}_{M}$ is a linear bounded operator, indeed,
$\left\Vert \boldsymbol{P}_{M}\right\Vert \leq1$.

We denote by $\boldsymbol{E}_{M}$ , $M\geq1$, the embedding $\mathcal{X}%
_{M}\rightarrow\mathcal{X}_{\infty}$. The following result is a consequence of
the above observations. If $\mathcal{Z}$, $\mathcal{Y}$ are real Banach
spaces, we denote by $\mathfrak{B}(\mathcal{Z},\mathcal{Y})$, the space of all
linear bounded operators from $\mathcal{Z}$ into $\mathcal{Y}$.

\begin{lemma}
\cite[Lemma 2]{Zuniga-Nonlinearity} \label{Lemma_Condition_A}With the above
notation, the following assertions hold true:

\noindent(i) $\mathcal{X}_{\infty}$, $\mathcal{X}_{M}$ for $M\geq1$, are
\ real Banach spaces, all with the norm $\left\Vert \cdot\right\Vert _{\infty}
$;

\noindent(ii) $\boldsymbol{P}_{M}\in\mathfrak{B}\left(  \mathcal{X}_{\infty
},\mathcal{X}_{M}\right)  $ and $\left\Vert \boldsymbol{P}_{M}\varphi
\right\Vert _{\infty}\leq\left\Vert \varphi\right\Vert _{\infty}$ for any
$M\geq1$, $\varphi\in\mathcal{X}_{\infty}$;

\noindent(iii) $\boldsymbol{E}_{M}\in\mathfrak{B}\left(  \mathcal{X}%
_{M},\mathcal{X}_{\infty}\right)  $ and $\left\Vert \boldsymbol{E}_{M}%
\varphi\right\Vert _{\infty}=\left\Vert \varphi\right\Vert _{\infty}$ for any
$M\geq1$, $\varphi\in\mathcal{X}_{M}$;

\noindent(iv) $\boldsymbol{P}_{M}\boldsymbol{E}_{M}\varphi=\varphi$ for
$M\geq1$, $\varphi\in\mathcal{X}_{M}$;

\noindent(v) $\lim_{M\rightarrow\infty}\left\Vert \varphi-\boldsymbol{P}%
_{M}\varphi\right\Vert _{\infty}=0$ for any $\varphi\in\mathcal{X}_{\infty}$;

\noindent(vi) $\lim_{M\rightarrow\infty}\Vert\boldsymbol{E}_{M}\boldsymbol{P}%
_{M}\phi-\phi\Vert_{\infty}=0$ for all $\phi\in\mathcal{X}_{\infty}$.
\end{lemma}

\begin{proposition}
\label{Prop2}Assume that $A(\Vert x\Vert_{p})$, $B(\Vert x-y\Vert_{p})$,
$U(x)$, $Z(x)\in\mathcal{X}_{M}$, $M\geq1$. Let $\tau$ be a fixed positive
real number. Consider the initial value problem:%
\begin{equation}
\left\{
\begin{array}
[c]{l}%
X\in C([0,\tau],\mathcal{X}_{M})\cap C^{1}([0,\tau],\mathcal{X}_{M})\\
\\%
\begin{array}
[c]{l}%
\frac{\partial X(x,t)}{\partial t}=-X(x,t)+%
{\displaystyle\int\limits_{\mathbb{Q}_{p}^{N}}}
A(\Vert x-y\Vert_{p})f(X(x,t))d^{N}y\\
+%
{\displaystyle\int\limits_{\mathbb{Q}_{p}^{N}}}
B(\Vert x-y\Vert_{p})U(y)d^{N}y+Z(x),\text{ }x\in B_{M}^{N}\text{, }t\geq0
\end{array}
\\
\\
X(x,0)=X_{0}\in\mathcal{X}_{M}.
\end{array}
\right.  \label{EQ_12}%
\end{equation}

There exists a unique $X\in C([0,\tau],\mathcal{X}_{M})$ which satisfies
\[
X(x,t)=e^{-t}X_{0}\left(  x\right)  +\int_{0}^{t}e^{-(t-s)}\boldsymbol{H}%
(X(x,s))ds
\]
where
\[
\boldsymbol{H}(X)(x,t)=%
{\displaystyle\int\limits_{\mathbb{Q}_{p}^{N}}}
A(\Vert x-y\Vert_{p})f(X(x,t))d^{N}y+%
{\displaystyle\int\limits_{\mathbb{Q}_{p}^{N}}}
B(\Vert x-y\Vert_{p})U(y)d^{N}y+Z(x).
\]
The function $X(x,t)$ is a solution of equation \ref{EQ_12} with initial datum
$X_{0}$.
\end{proposition}

\begin{proof}
The result is established by using the argument given in the proof of Theorem
\ref{Theorem0}.
\end{proof}

By the discussion presented in section \ref{Section_Discretization},
(\ref{EQ_12}) describes a $p$-adic discrete CNN. Furthermore, Theorem
\ref{Theorem1} is also valid for discrete CNN in $\mathcal{X}_{M}$.

\begin{remark}
\label{Nota2}By using the discretization procedure given in Section
\ref{Section_Discretization} and in the proof of Theorem \ref{Theorem0},
Proposition \ref{Prop2} implies that the initial value problem
\[
\left\{
\begin{array}
[c]{l}%
X_{M}\in C([0,\tau],\mathcal{X}_{M})\cap C^{1}([0,\tau],\mathcal{X}_{M})\\
\\
\frac{\partial X_{M}}{\partial t}=-X_{M}+\boldsymbol{P}_{M}\boldsymbol{H}%
(\boldsymbol{E}_{M}X_{M})\\
\\
X_{M}(0)=\boldsymbol{P}_{M}(X_{0})
\end{array}
\right.
\]
has a unique solution for an arbitrary $\tau>0$.
\end{remark}

\begin{theorem}
\label{Theorem1A}All the states $X(\boldsymbol{i},t)$, $\boldsymbol{i}\in
G_{M}^{N}$, in a $p$-adic discrete CNN are bounded for all time $t\geq0$. More
precisely, if
\begin{multline*}
X_{\max}:=\max_{\boldsymbol{i}\in G_{M}^{N}}\left\vert X_{0}(\boldsymbol{i}%
)\right\vert +p^{-MN}\left(  \max_{\boldsymbol{i}\in G_{M}^{N}}\left\vert
f(\boldsymbol{i})\right\vert \right)
{\displaystyle\sum\limits_{\boldsymbol{i}\in G_{M}^{N}}}
\left\vert A\left(  \boldsymbol{i}\right)  \right\vert \\
+p^{-MN}\left(  \max_{\boldsymbol{i}\in G_{M}^{N}}\left\vert U(\boldsymbol{i}%
)\right\vert \right)
{\displaystyle\sum\limits_{\boldsymbol{i}\in G_{M}^{N}}}
\left\vert A\left(  \boldsymbol{i}\right)  \right\vert +\max_{\boldsymbol{i}%
\in G_{M}^{N}}\left\vert Z(\boldsymbol{i})\right\vert ,
\end{multline*}
then
\[
|X(\boldsymbol{i},t)|\leq X_{\max}\text{ for all }t\geq0\text{ and for all
}\boldsymbol{i}\in G_{M}^{N}.
\]

In addition%
\[
X_{-}\left(  \boldsymbol{i}\right)  :=\lim\inf_{t\rightarrow\infty
}X(\boldsymbol{i},t)\leq X(\boldsymbol{i},t)\leq\lim\sup_{t\rightarrow\infty
}X(\boldsymbol{i},t)=:X_{+}\left(  \boldsymbol{i}\right)  ,
\]
for $\boldsymbol{i}\in G_{M}^{N}$. If $X_{-}\left(  \boldsymbol{i}\right)
=X_{+}\left(  \boldsymbol{i}\right)  :=X^{\ast}(\boldsymbol{i})$, then
\begin{align*}
X^{\ast}(\boldsymbol{i})  &  =%
{\displaystyle\sum\limits_{\boldsymbol{j}\in G_{M}^{N}}}
p^{-MN}A(\Vert\boldsymbol{i}-\boldsymbol{j}\Vert_{p})f(X^{\ast}(\boldsymbol{i}%
))\\
&  +%
{\displaystyle\sum\limits_{\boldsymbol{j}\in G_{M}^{N}}}
p^{-MN}B(\Vert\boldsymbol{i}-\boldsymbol{j}\Vert_{p})U(\boldsymbol{j}%
)+Z(\boldsymbol{i})\text{, for }\boldsymbol{i}\in G_{M}^{N}\text{,}%
\end{align*}
and
\begin{multline*}
X^{\ast}(\boldsymbol{i})\geq-p^{-MN}\left(  \max_{\boldsymbol{i}\in G_{M}^{N}%
}\left\vert f(\boldsymbol{i})\right\vert \right)
{\displaystyle\sum\limits_{\boldsymbol{i}\in G_{M}^{N}}}
\left\vert A\left(  \boldsymbol{i}\right)  \right\vert \\
-p^{-MN}\left(  \max_{\boldsymbol{i}\in G_{M}^{N}}\left\vert U(\boldsymbol{i}%
)\right\vert \right)
{\displaystyle\sum\limits_{\boldsymbol{i}\in G_{M}^{N}}}
\left\vert A\left(  \boldsymbol{i}\right)  \right\vert +Z(\boldsymbol{i}%
)\text{ , for all }\boldsymbol{i}\in G_{M}^{N}.
\end{multline*}

\end{theorem}

\begin{theorem}
\label{Theorem2} Let $X$ be the solution of a continuous $p$-adic CNN given by
Theorem \ref{Prop1} with initial condition $X_{0}$. Let $X_{M}$ be the
solution of the Cauchy problem
\begin{equation}
\left\{
\begin{array}
[c]{l}%
\frac{dX_{M}}{dt}=-X_{M}+\boldsymbol{P}_{M}\boldsymbol{H}(\boldsymbol{E}%
_{M}X_{M})\\
X_{M}(0)=\boldsymbol{P}_{M}(X_{0}),
\end{array}
\right.  \label{eq:discrete cauchy}%
\end{equation}
cf. Proposition \ref{Prop2} and Remark \ref{Nota2}. Then
\[
\lim_{M\rightarrow\infty}\sup_{0\leq t\leq\tau}\Vert X_{M}(t)-X(t)\Vert
_{\infty}=0.
\]

\end{theorem}

\begin{proof}
The result follows from Lemma \ref{Lemma_Condition_A}, Propositions
\ref{Prop1}, \ref{Prop2}, by using standard techniques of approximation for
evolution equations, see e.g. \cite[Theorem 5.4.7]{Milan}. See also
\cite[Section 9.1 and Theorem 7]{Zuniga-Nonlinearity} for an in-depth
discussion of similar matters.
\end{proof}

\section{Numerical Simulations of $p$-Adic Continuous CNNs}

In this section we present some numerical simulations of the solutions of
several $p$-adic continuous $\text{CNN}$s in dimension $1$. We give two
numerical schemes for the numerical approximation of the solutions.

\subsection{Numerical Scheme A}

\begin{lemma}
\label{lemma: approx radial} Let $H(|\cdot|_{p})\in L^{1}(\mathbb{Q}_{p})$ and
let $g\in\mathcal{X}_{\infty}$. We set $G_{k}=p^{-k}\mathbb{Z}_{p}%
/p^{k}\mathbb{Z}_{p}$, $k\in\mathbb{N}$. Then {\small
\[%
{\displaystyle\int\limits_{\mathbb{Q}_{p}}}
H(|x-y|)g(y)dy=\lim_{k\rightarrow\infty}\sum_{\boldsymbol{i}\in G_{k};\text{
}\boldsymbol{i}\neq x}g(\boldsymbol{i})p^{-k}H(|x-\boldsymbol{i}%
|_{p})+g(x)(1-p^{-1})\sum_{l=k}^{\infty}H(p^{-l})p^{-l}.
\]
}
\end{lemma}

\begin{proof}
By Lemma \ref{Lemma_Condition_A}-(v), $\lim_{k\rightarrow\infty}%
\sum_{\boldsymbol{i}\in G_{k}}g(\boldsymbol{i})\Omega(p^{k}|x-\boldsymbol{i}%
|_{p})=g(x)$, now by the dominated convergence theorem,
\begin{align*}%
{\displaystyle\int\limits_{\mathbb{Q}_{p}}}
H(|x-y|_{p})g(y)dy  &  =\lim_{k\rightarrow\infty}\sum_{\boldsymbol{i}\in
G_{k}}g(\boldsymbol{i})%
{\displaystyle\int\limits_{\mathbb{Q}_{p}}}
H(|x-y|_{p})\Omega(p^{k}|y-\boldsymbol{i}|_{p})dy\\
&  =\lim_{k\rightarrow\infty}\sum_{\boldsymbol{i}\in G_{k}}g(\boldsymbol{i})%
{\displaystyle\int\limits_{x-\boldsymbol{i}+p^{k}\mathbb{Z}_{p}}}
H(|z|_{p})dz.
\end{align*}
Now, if $|x-\boldsymbol{i}|_{p}>p^{-k}$, i.e. $x\neq\boldsymbol{i}$ \ in
$G_{k}$, then
\[%
{\displaystyle\int\limits_{x-\boldsymbol{i}+p^{k}\mathbb{Z}_{p}}}
H(|z|_{p})dz=p^{-k}H(|x-\boldsymbol{i}|_{p}).
\]
And if $|x-\boldsymbol{i}|_{p}\leq p^{-k}$, i.e. $x=\boldsymbol{i}$ \ in
$G_{k}$, then
\[%
{\displaystyle\int\limits_{x-\boldsymbol{i}+p^{k}\mathbb{Z}_{p}}}
H(|z|_{p})dz=\sum_{l=k}^{\infty}H(p^{-l})(1-p^{-1})p^{-l}.
\]

\end{proof}

We now assume that $A$, $B$ are radial integrable functions, \ and that $U$,
$Z$, $X_{0}\in\mathcal{X}_{\infty}$. \ Based on the continuity of operators
$\boldsymbol{A},\boldsymbol{B}:\mathcal{X}_{\infty}\rightarrow\mathcal{X}%
_{\infty}$ and the formula given in Lemma \ref{lemma: approx radial}, we can
approximate the solution $X(x,t)$ of a $p$-adic continuous CNN$\left(
A,B,U,Z\right)  $ by $p^{2k}$ ODEs, $k\geq1$, of the form%
\begin{multline*}
\frac{d}{dx}X(\boldsymbol{i},t)=-X(\boldsymbol{i},t)+\sum_{\boldsymbol{j}\in
G_{k};\text{ }\boldsymbol{j}\neq\boldsymbol{i}}f(X(\boldsymbol{j}%
,t))p^{-k}A(|\boldsymbol{i}-\boldsymbol{j}|_{p})+\\
f(X(\boldsymbol{i},t))(1-p^{-1})\sum_{l=k}^{k_{\max}}A(p^{-l})p^{-l}%
+\sum_{\boldsymbol{j}\in G_{k};\text{ }\boldsymbol{j}\neq\boldsymbol{i}%
}U(\boldsymbol{i})p^{-k}B(|\boldsymbol{i}-\boldsymbol{j}|_{p})\\
+U(\boldsymbol{j})(1-p^{-1})\sum_{l=k}^{k_{\max}}B(p^{-l})p^{-l}%
+Z(\boldsymbol{i})\text{, for }\boldsymbol{i}\in G_{k}.
\end{multline*}
In the simulations the parameters $k$, $k_{max}$ were chosen by trial and
error on a case by case approach. The sum $\sum_{l=k}^{k_{\max}}%
A(p^{-l})p^{-l}$ can be approximated by $A(p^{-k})p^{-k}$ in the cases were
$A(p^{-k})p^{-k}$ is the dominant term in $\sum_{l=k}^{k_{\max}}$
$A(p^{-l})p^{-l}$.

\subsection{Numerical Scheme B}

\begin{lemma}
\label{lemma: approx test} Let $H(x)=\sum_{l=0}^{m}H_{l}\Omega(p^{k_{l}%
}|x-b_{l}|_{p})$ be a test function and let $g\in\mathcal{X}_{\infty}$. Take
$G_{k}=p^{-k}\mathbb{Z}_{p}/p^{k}\mathbb{Z}_{p}$, $k\in\mathbb{N}$, as before.
Then
\begin{gather*}%
{\displaystyle\int\limits_{\mathbb{Q}_{p}}}
H(x-y)g(y)dy=\lim_{k\rightarrow\infty}\sum_{\boldsymbol{i}\in G_{k}%
}g(\boldsymbol{i})\sum_{l=0}^{m}H_{l}%
{\displaystyle\int\limits_{\mathbb{Q}_{p}}}
\Omega(p^{k_{l}}|x-\boldsymbol{i}-b_{l}-y|_{p})\Omega(p^{k}|y|_{p})dy\\
=\lim_{k\rightarrow\infty}\sum_{\boldsymbol{i}\in G_{k}}g(\boldsymbol{i}%
)\sum_{l=0}^{m}H_{l}p^{\min(-k,-k_{l})}\Omega\left(  p^{-\max(-k,-k_{l}%
)}|x-\boldsymbol{i}-b_{l}|_{p}\right) \\
=\lim_{k\rightarrow\infty}\sum_{\boldsymbol{i}\in G_{k}}g(\boldsymbol{i}%
)\sum_{l=0}^{m}H_{l}p^{-\max(k,k_{l})}\Omega\left(  p^{\min(k,k_{l}%
)}|x-\boldsymbol{i}-b_{l}|_{p}\right)  .
\end{gather*}

\end{lemma}

\begin{proof}
It is sufficient to consider the case where $H(x)=\Omega(p^{k_{H}}%
|x-b_{H}|_{p})$ for some $k_{H}\in\mathbb{Z}$ and $b_{H}\in\mathbb{Q}_{p}$.
Since $g(x)=\lim_{k\rightarrow\infty}\sum_{\boldsymbol{i}\in G_{k}%
}g(\boldsymbol{i})\Omega(p^{k}|x-\boldsymbol{a}|_{p})$, we have
\begin{align*}%
{\displaystyle\int\limits_{\mathbb{Q}_{p}}}
H(x-y)g(y)dy  &  =\lim_{k\rightarrow\infty}\sum_{\boldsymbol{i}\in G_{k}%
}g(\boldsymbol{i})%
{\displaystyle\int\limits_{\mathbb{Q}_{p}}}
\Omega(p^{k_{H}}|x-b_{H}-y|_{p})\Omega(p^{k}|y-\boldsymbol{i}|_{p})dy\\
&  =\lim_{k\rightarrow\infty}\sum_{\boldsymbol{i}\in G_{k}}g(\boldsymbol{i})%
{\displaystyle\int\limits_{\mathbb{Q}_{p}}}
\Omega(p^{k_{H}}|(x-b_{H}-\boldsymbol{i})-y|_{p})\Omega(p^{k}|y|_{p})dy.
\end{align*}
Without loss of generality, we may assume that $k_{H}\leq k$, and since any
two balls are disjoint or one contains the other, then $B_{-k}\cap B_{-k_{H}%
}(x-b_{H}-a)=\emptyset$ or $B_{-k}\cap B_{-k_{H}}(x-b_{H}-\boldsymbol{i}%
)=B_{-k}$. The latter case occurs if and only if $0\in B_{-k_{H}}%
(x-b_{H}-\boldsymbol{i})$, i.e. when $|x-b_{H}-a|_{p}\leq p^{-k_{H}}$.
Therefore
\[%
{\displaystyle\int\limits_{\mathbb{Q}_{p}}}
\Omega(p^{k_{H}}|x-b_{H}-\boldsymbol{i}-y|_{p})\Omega(p^{k}|y|_{p}%
)dy=p^{-k}\Omega\left(  p^{k_{H}}|x-b_{H}-\boldsymbol{i}|_{p}\right)  .
\]

\end{proof}

We now assume that $U$, $Z$, $X_{0}\in\mathcal{X}_{\infty}$ and that $A$, $B$
are test functions of the form%
\[
A(x)=\sum_{l=0}^{m_{A}}A_{l}\Omega(p^{k_{l}}|x-a_{l}|_{p}),\text{ \ }%
B(x)=\sum_{l=0}^{m_{B}}B_{l}\Omega(p^{k_{l}}|x-b_{l}|_{p}.
\]
Based on the continuity of operators $\boldsymbol{A},\boldsymbol{B}%
:\mathcal{X}_{\infty}\rightarrow\mathcal{X}_{\infty}$ and the formula given in
Lemma \ref{lemma: approx test}, we can approximate the solution $X(x,t)$ of a
$p$-adic continuous CNN by $p^{2k}$ ODEs, $k\geq1$, of the form%
\begin{gather*}
\frac{d}{dx}X(\boldsymbol{i},t)=-X(\boldsymbol{i},t)+\sum_{\boldsymbol{j}\in
G_{k}}f(X(\boldsymbol{j},t))\sum_{l=0}^{m_{A}}A_{l}p^{-\max(k,k_{l})}%
\Omega\left(  p^{\min(k,k_{l})}|\boldsymbol{i}-\boldsymbol{j}-a_{l}%
|_{p}\right) \\
+\sum_{\boldsymbol{j}\in G_{k}}U(\boldsymbol{j)}\sum_{l=0}^{m_{B}}%
B_{l}p^{-\max(k,k_{l})}\Omega\left(  p^{\min(k,k_{l})}|\boldsymbol{i}%
-\boldsymbol{j}-b_{l}|_{p}\right)  +Z(\boldsymbol{i})\text{, for
}\boldsymbol{i}\in G_{k}.
\end{gather*}
It is possible to combine the approximations given in numeric schemes A, B.

\subsection{A remark on the visualization of finite rooted trees}

The discretizations of \ the kernels $A$, $B$ are functions on $G_{k}\times
G_{k}$, while the input $U$ and $X_{0}$ are functions on $G_{k}$. We use
systematically heat maps to present these functions. We always include a plot
of the tree $G_{k}$. By convention we identify the leaves of the tree $G_{k}$
with the set of rational numbers $\{0,1/p^{k},2/p^{k},\ldots,(p^{2k}%
-1)/p^{k}\}$. Furthermore, \ we label the levels of $G_{k}$ with integers from
the set $\left\{  -k,-k+1,\ldots,0,1,\ldots k-1\right\}  $. The level $l$
consists of the cells $\boldsymbol{i}$, $\boldsymbol{j}$ such that
\[
-\log_{p}(|\boldsymbol{i}-\boldsymbol{j}|_{p})=(\text{the level of the first
common ancestor of }\boldsymbol{i},\boldsymbol{j})=l.
\]

\subsection{First Simulation}

In this example, we take $k=2$, $p=2$, which means that we use a tree with
$2^{4}=16$ leaves and $4$ levels. A basic application of the classical CNNs is
image processing, see e.g. \cite{Chua-Tamas}. In this example we present a
one-dimensional edge detector, which is a $p$-adic, one-dimensional analog of
the examples 3.1 and 3.2 in \cite{Chua-Tamas}. The input $U$ is a image having
three levels:%
\[
U(x)=\sum_{\boldsymbol{i}\in G_{2}}U_{\boldsymbol{i}}\Omega(2^{2}%
|x-\boldsymbol{i}|_{2}),\;\;U_{\boldsymbol{i}}=\left\{
\begin{array}
[c]{lll}%
-1 & \text{if} & \boldsymbol{i}=1,2,1/4,13/4\\
0 & \text{if} & \boldsymbol{i}=1/2,9/4,5/4\\
1 & \text{otherwise,} &
\end{array}
\right.
\]
$x\in G_{2}=2^{-2}\mathbb{Z}_{2}/2^{2}\mathbb{Z}_{2}$. As in \cite{Chua-Tamas}
we take $X_{0}(x)=0$,$\ A(x)=0$. To construct template $B$, we identify a
matrix with a test function. We use%
\[
B(x)=64\Omega(2^{2}|x|_{2})-4%
{\displaystyle\sum\limits_{\boldsymbol{i}\in G_{2}\text{; }\boldsymbol{i}%
\neq0}}
\Omega(2^{2}|x-\boldsymbol{i}|_{2}),\text{ }x\in G_{2}\text{.}%
\]
Finally, we take $Z(x)=-\Omega(2^{-2}|x|_{2}),$ \ \ $f(x)=\frac{1}{2}\left(
|x+1|-|x-1|\right)  $. The output $Y(x,t)$ consists of the edges on the input
$U$, see Figure \ref{SIM_1_FIG_3}.%
\begin{figure}
[h]
\begin{center}
\includegraphics[
height=2.968in,
width=5.0246in
]%
{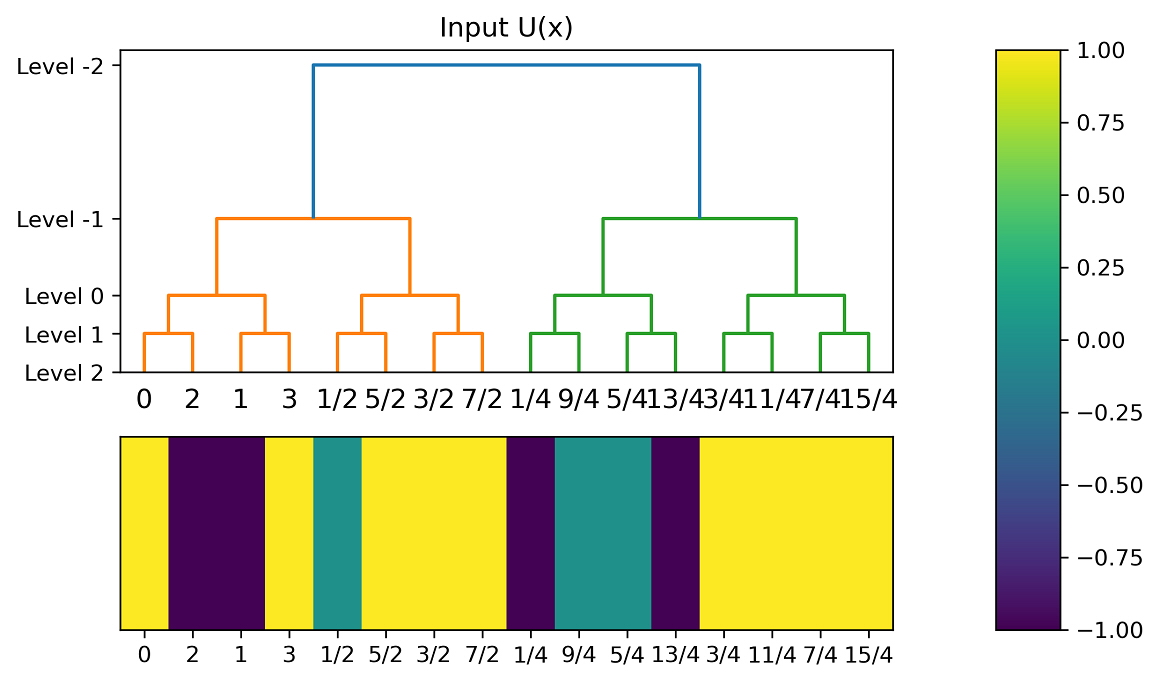}%
\caption{Simulation 1. Heat map $U(x)$.}%
\end{center}
\end{figure}
%

\begin{figure}
[h]
\begin{center}
\includegraphics[
height=2.9663in,
width=3.2422in
]%
{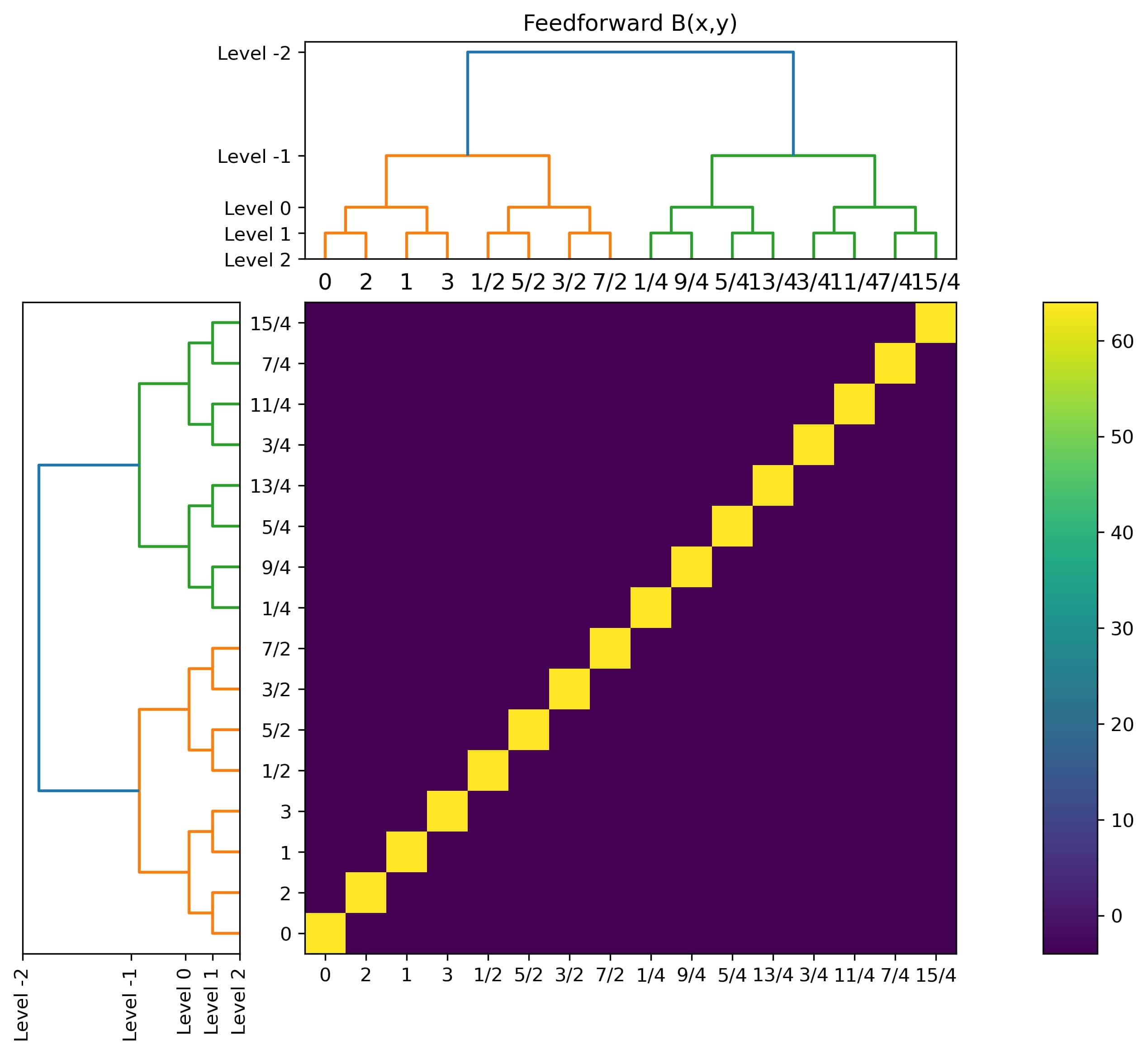}%
\caption{Simulation $1$. Heat map of $B\left(  \left\vert x-y\right\vert
_{2}\right)  $, $x$, $y\in G_{2}$.}%
\end{center}
\end{figure}
\begin{figure}
[hh]
\begin{center}
\includegraphics[
height=2.9672in,
width=4.2791in
]%
{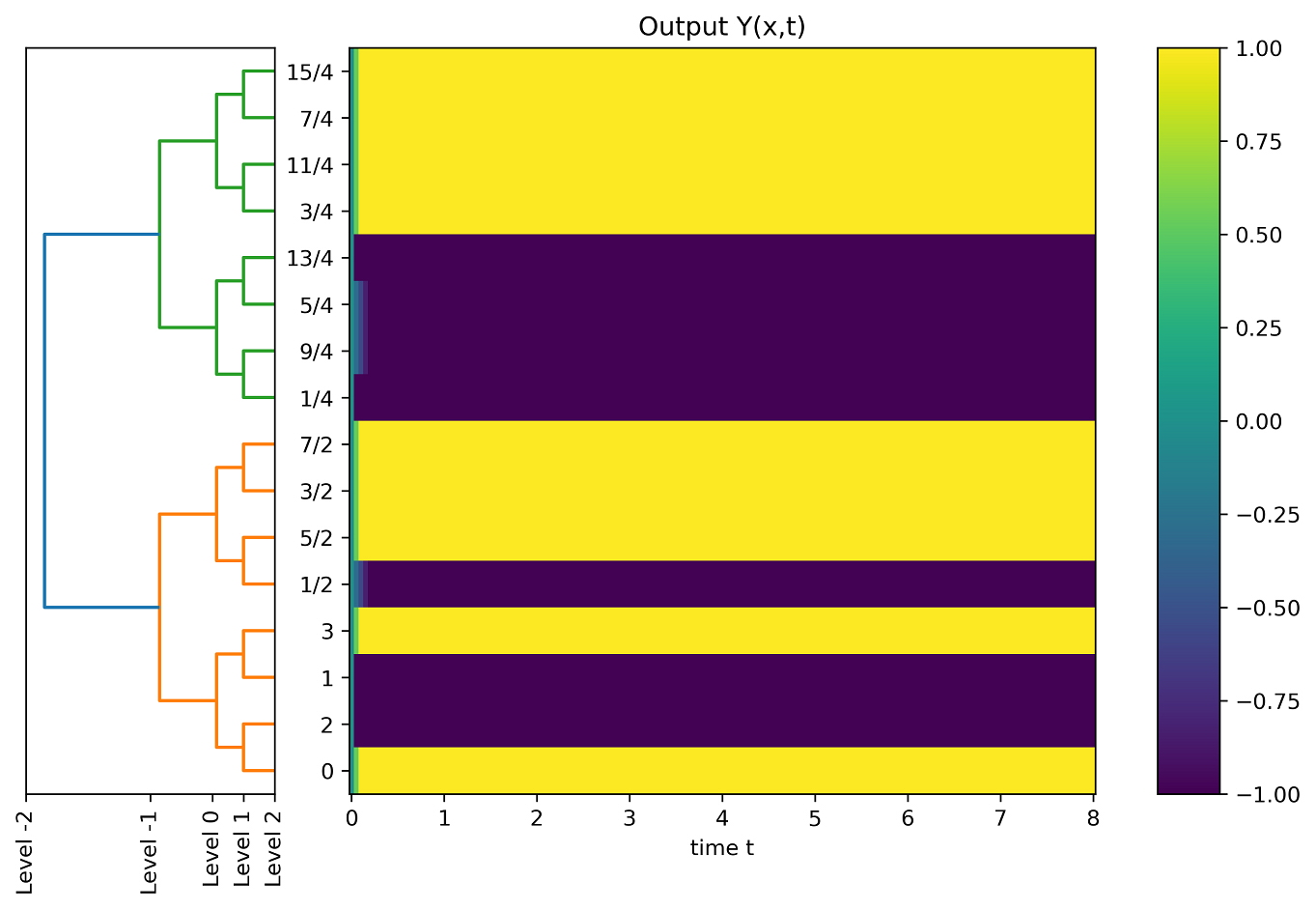}%
\caption{Simulation $1$. Step $0.05$.}%
\label{SIM_1_FIG_3}%
\end{center}
\end{figure}
\newpage

\subsection{Second Simulation}

In this example, we take $k=2$, $p=2$, which means that we use a tree with
$2^{4}=16$ leaves and $4$ levels. We consider a CNN with the followin
parameters:%
\[
A(x)=\Omega(2^{2}|x-2^{-2}|_{2}),\text{ \ }B(x)=U(x)=\Omega(2^{2}%
|x|_{2}),\text{ }Z(x)=0,\text{ }x\in G_{2.}\text{\ }%
\]
We set $X_{0}(x)=0$ and \ \ $f(x)=\frac{1}{2}\left(  |x+1|-|x-1|\right)  $.%

\begin{figure}
[h]
\begin{center}
\includegraphics[
height=2.9663in,
width=3.2422in
]%
{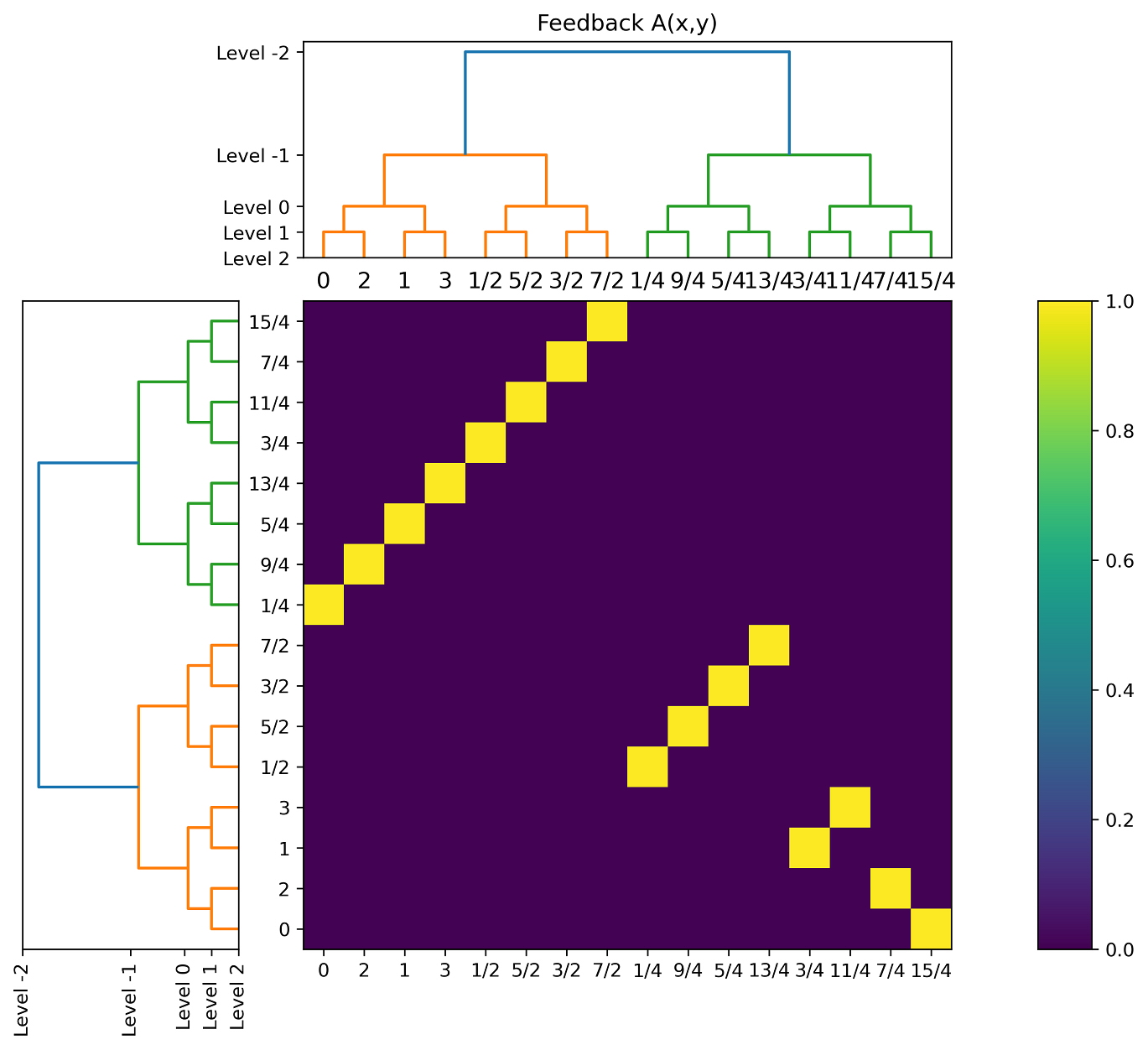}%
\caption{Simulation $2$. Heat map $A(x-y)$ for $x,y\in G_{2}$.}%
\end{center}
\end{figure}

In this network, we have $A(\boldsymbol{i},\boldsymbol{j})=A(\boldsymbol{i}%
-\boldsymbol{j})=\Omega\left(  2^{2}\left\vert \boldsymbol{i}-\boldsymbol{j}%
-2^{-1}\right\vert _{2}\right)  $, $B(\boldsymbol{i},\boldsymbol{j}%
)=B(\left\vert \boldsymbol{i}-\boldsymbol{j}\right\vert _{2})=\Omega\left(
2^{2}\left\vert \boldsymbol{i}-\boldsymbol{j}\right\vert _{2}\right)
=\delta_{\boldsymbol{i},\boldsymbol{j}}$, where $\delta_{\boldsymbol{i}%
,\boldsymbol{j}}$\ denotes the Konecker delta function. This network does not
have the space-invariant property because $A(\boldsymbol{i},\boldsymbol{j}%
)=\Omega\left(  2^{2}\left\vert \boldsymbol{i}-\boldsymbol{j}-2^{-1}%
\right\vert _{2}\right)  $ is not a radial
function. Due to this fact, $A(\boldsymbol{i},\boldsymbol{j})$ is not a
symmetric matrix. For instance:%
\[
A(\frac{15}{4},0)=0,\text{ }A(0,\frac{15}{4})=1,\text{ \ }A(\frac{1}%
{4},0)=0,\text{ }A(0,\frac{1}{4})=1.
\]
%

\begin{figure}
[h]
\begin{center}
\includegraphics[
height=2.9663in,
width=3.2422in
]%
{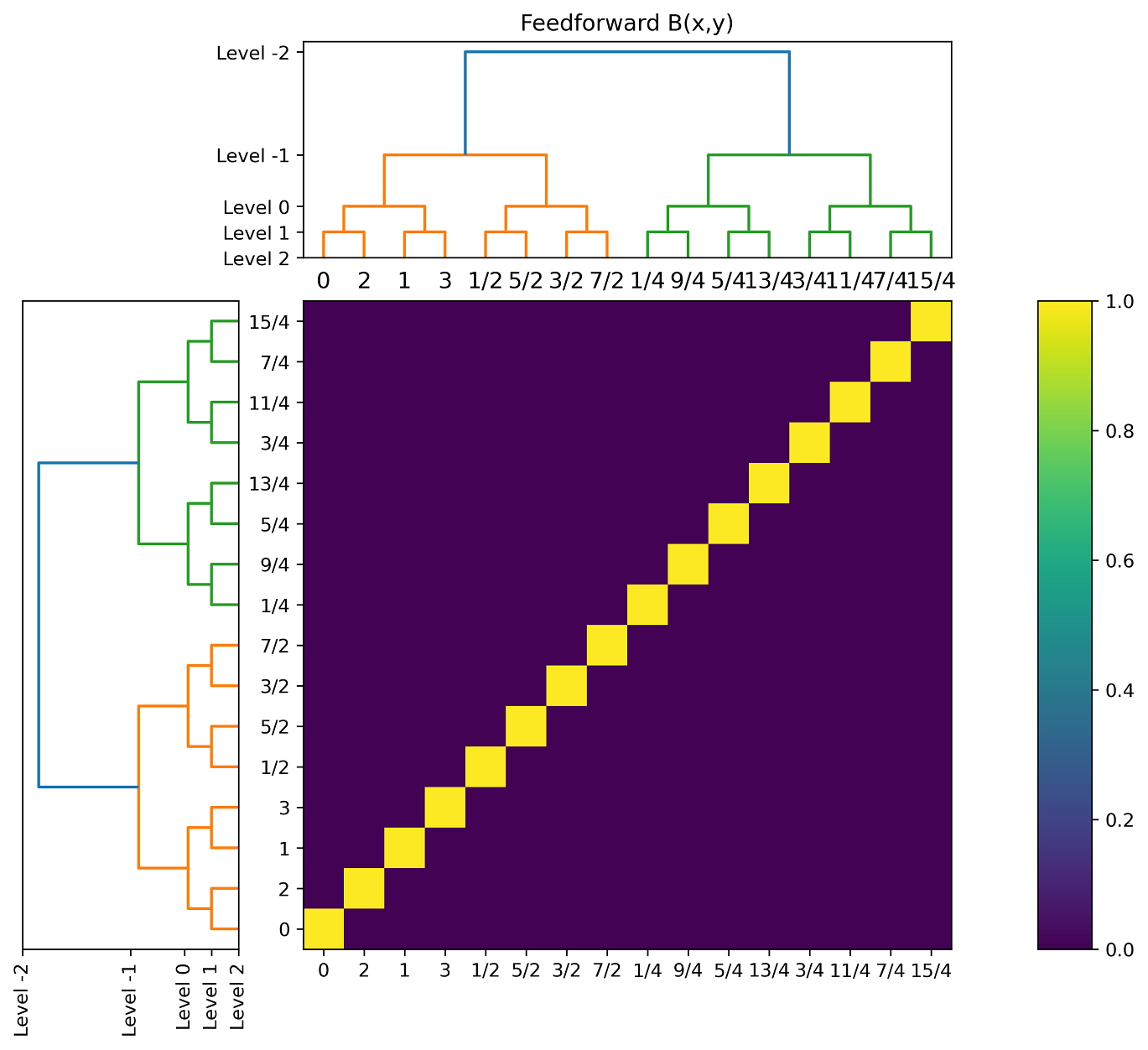}%
\caption{Simulation $2$. \ Heat map of $B\left(  \left\vert x-y\right\vert
_{2}\right)  $ for $x,y\in G_{2}$.}%
\end{center}
\end{figure}
%

\begin{figure}
[h]
\begin{center}
\includegraphics[
height=3.1393in,
width=5.1733in
]%
{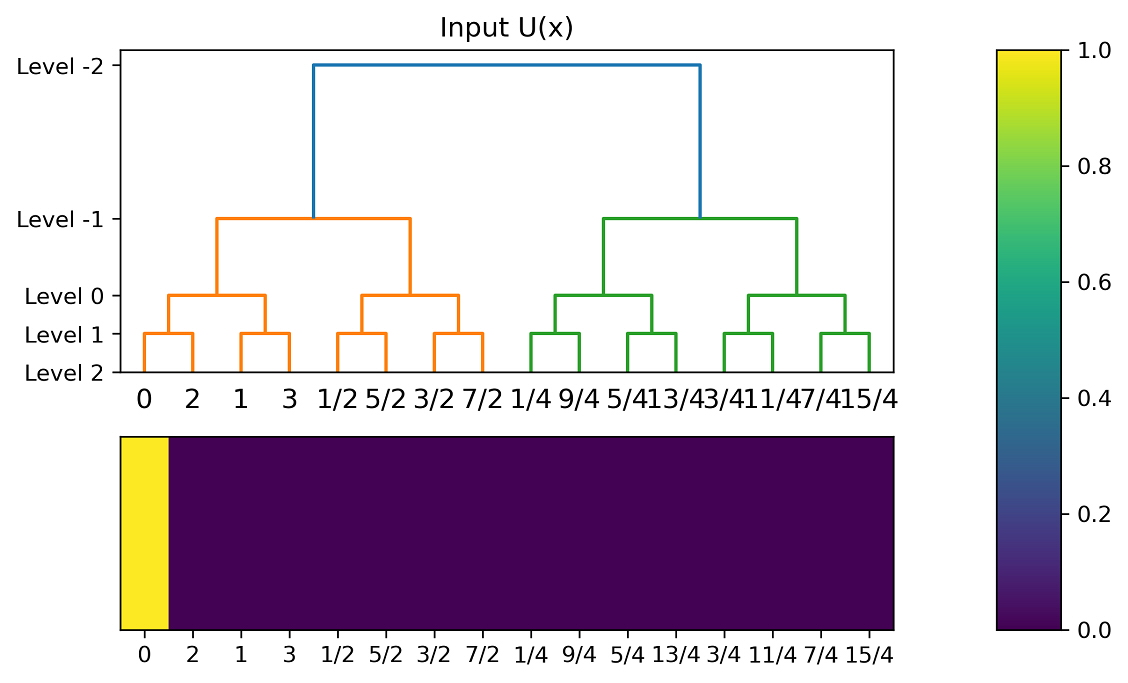}%
\caption{Simulation $2$. Heat map of $U(x)$.}%
\end{center}
\end{figure}

Our interpretation is that there is a connection from cell $\frac{15}{4}$ to
cell $0$, \ and a connection from cell $0$ to cell $\frac{1}{4}$. This
assertion is confirm by the ouput $Y(x,t)$, see Figure \ref{SIM_2_FIg_4}.
Notice that $Y(\frac{1}{2},t)\neq0$ and $A(\frac{1}{2},0)=A(0,\frac{1}{2})=0$.
But $A(\frac{1}{4},\frac{1}{2})=0$, $A(\frac{1}{2},\frac{1}{4})=1$, then there
is a connection from cell $\frac{1}{4}$ to cell $\frac{1}{2\text{ }}$, which
explains the fact that $Y(\frac{1}{2},t)\neq0$.

\newpage

The numerical solutions is given in Figure \ref{SIM_2_FIg_4}. We now take
$A(x)=B(x)=\Omega(2^{2}|x|_{2})$. In this case the output $Y(x,t)$ changes
completely, see Figure \ref{SIM_2_FIG_5}.%

\begin{figure}
[h]
\begin{center}
\includegraphics[
height=2.9672in,
width=4.2203in
]%
{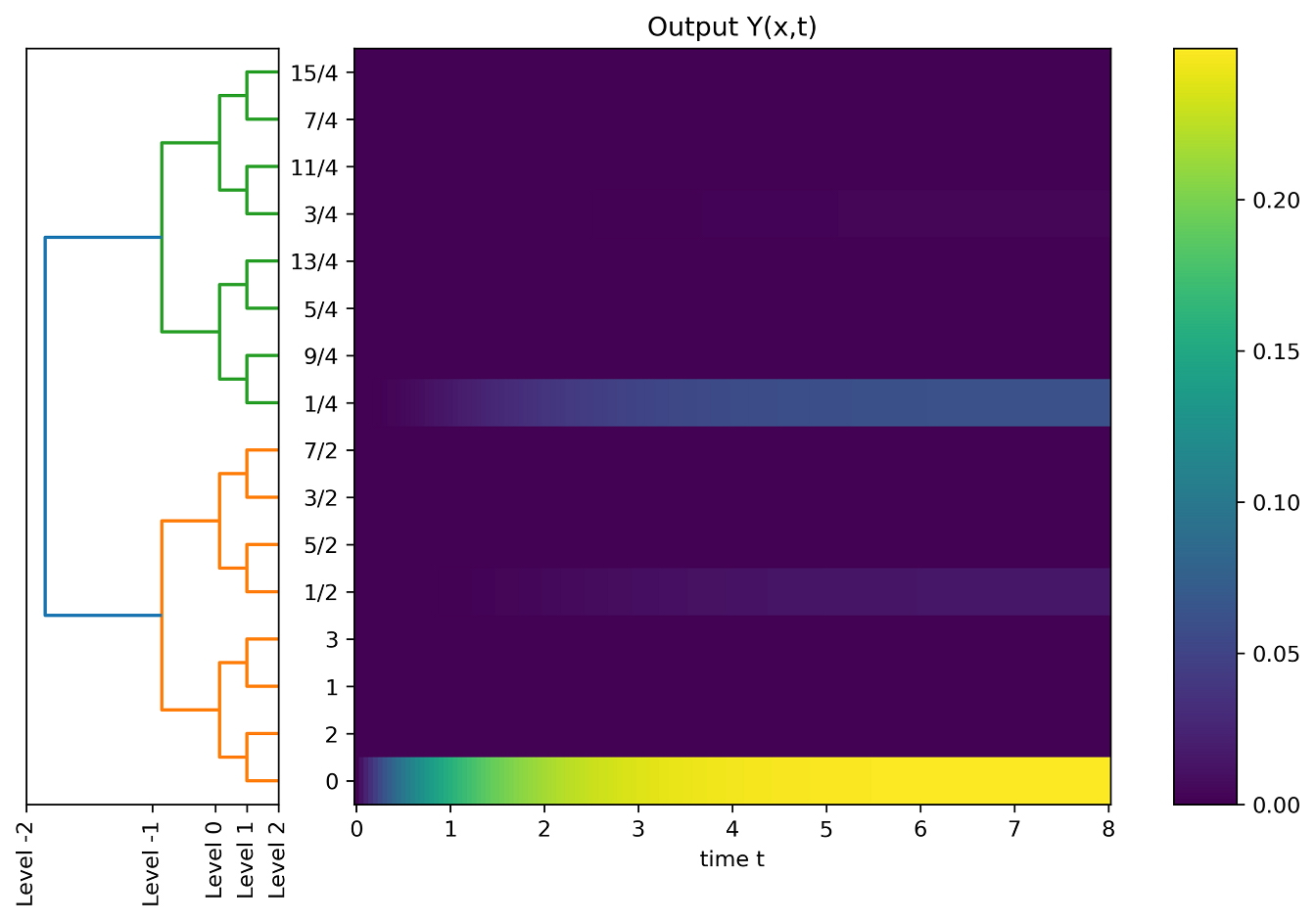}%
\caption{Simulation $2$. Step $0.05$.}%
\label{SIM_2_FIg_4}%
\end{center}
\end{figure}
%

\begin{figure}
[h]
\begin{center}
\includegraphics[
height=2.9672in,
width=4.2203in
]%
{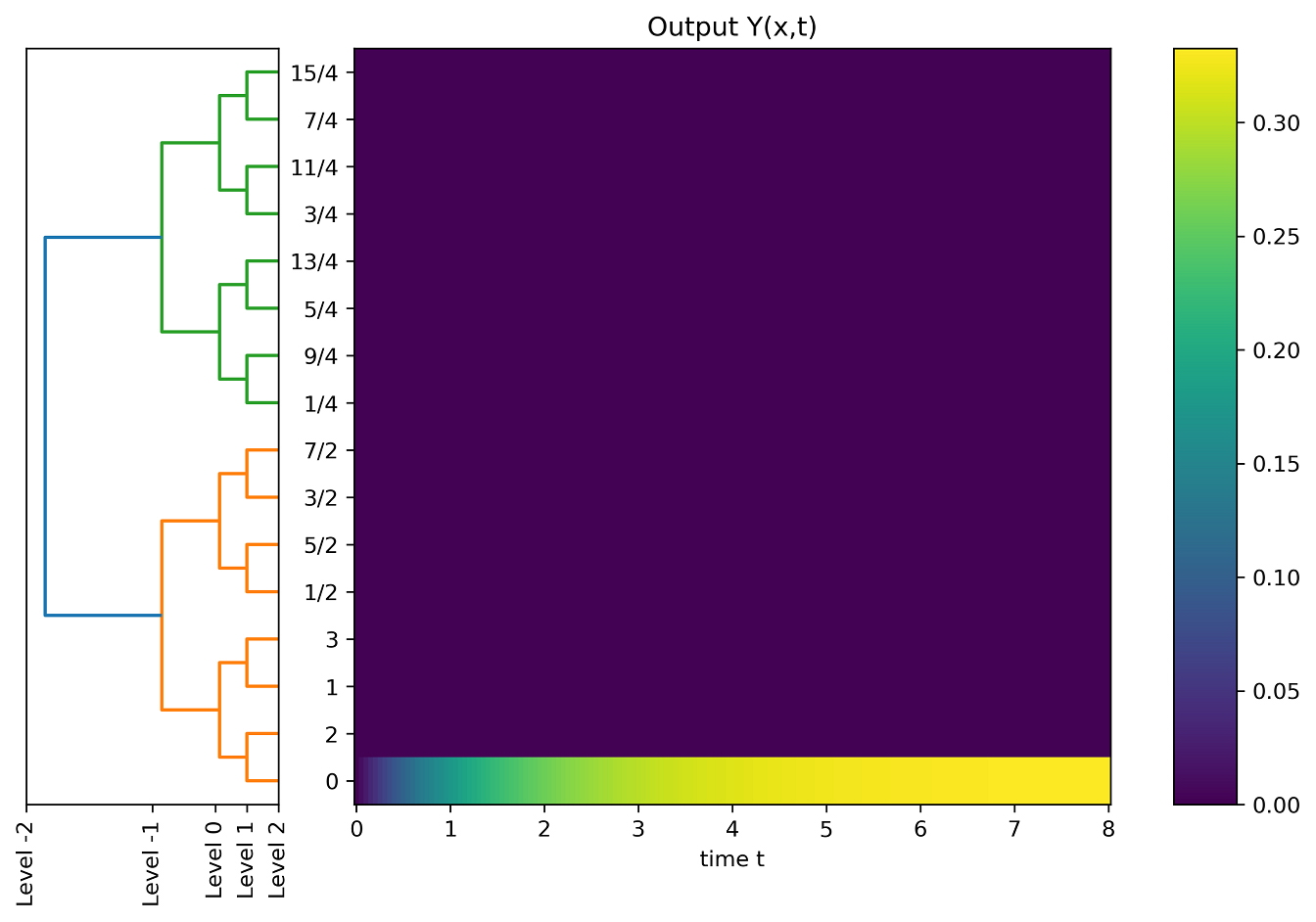}%
\caption{Simulation 2. Output with \ $A(x)=B(x)=\Omega(2^{2}|x|_{2})$ and step
$0.05$.}%
\label{SIM_2_FIG_5}%
\end{center}
\end{figure}
\newpage

\subsection{Third Simulation}

In this example, we take $k=3$, $p=2$, which means that we use a tree with
$2^{6}=64$ leaves and $12$ levels. We consider a CNN with the following
parameters: $A(x)=\Omega(2^{3}|x-2^{-2}|_{2})$, $B(|x|_{2})=\Omega
(2^{3}|x|_{2})$, $U(x)=\sin(p^{4}|x|_{2})$,$\ Z(x)=0.15\Omega(2^{-2}|x|_{2})$
for $x\in G_{3}=2^{-3}\mathbb{Z}_{3}/2^{3}\mathbb{Z}_{3}$. We set $X_{0}(x)=0$
and \ \ $f(x)=\frac{1}{2}\left(  |x+1|-|x-1|\right)  $.

As a consequence of the fractal nature of the $p$-adic numbers, the $p$-adic
CNNs exhibit self-similarity in several ways. For instance, the graph of the
kernel $A(x,y)$ is a self-similar set, this follows by comparing the graphs
given in simulations $2$ and $3$ for this kernel. In addition, the output
$Y(x,t)=0$ when the norm $|x|_{2}$\ is sufficiently large. In this simulation
the CNN produces a pattern similar to the input.%

\begin{figure}
[h]
\begin{center}
\includegraphics[
height=2.2321in,
width=2.5668in
]%
{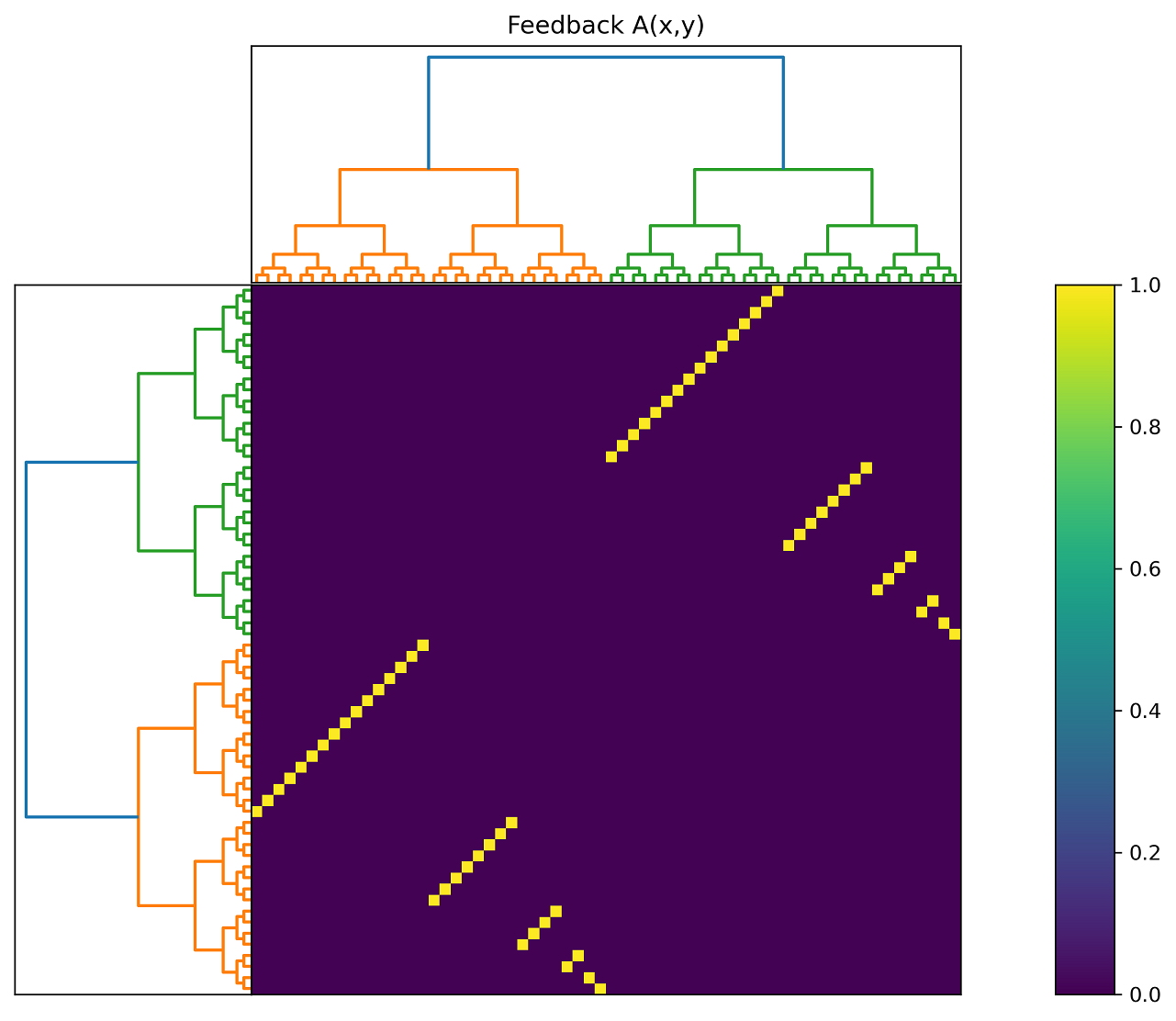}%
\caption{Simulation 3. Heat map of $A(x-y)$ for $x$, $y\in G_{3}$.}%
\end{center}
\end{figure}
%

\begin{figure}
[ptb]
\begin{center}
\includegraphics[
height=2.2329in,
width=3.8156in
]%
{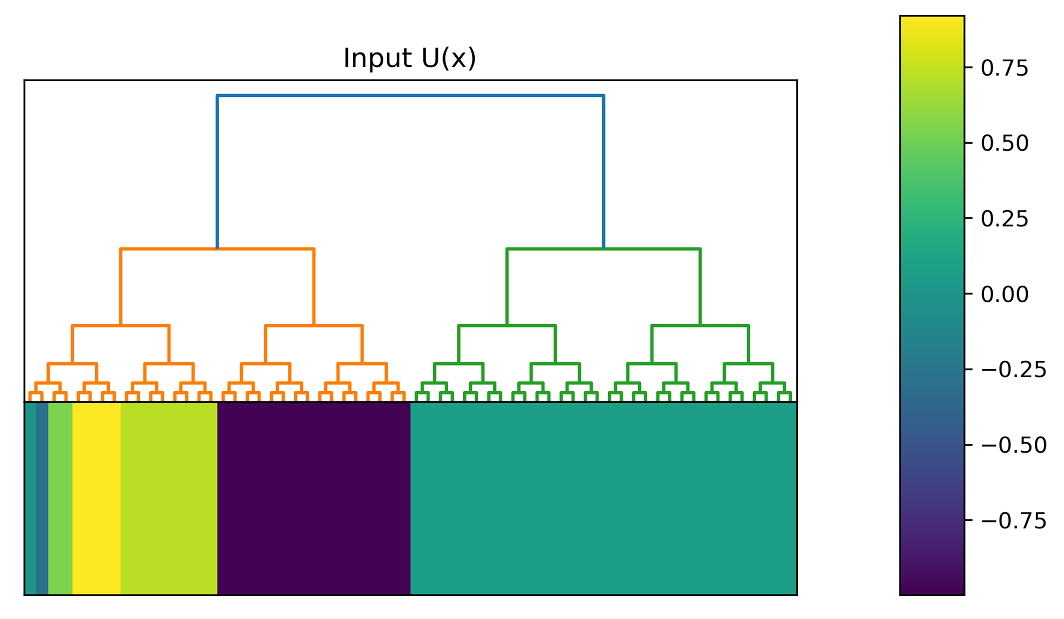}%
\caption{Simulation $3$. Heat map of $U(x)$.}%
\end{center}
\end{figure}
%

\begin{figure}
[h]
\begin{center}
\includegraphics[
height=2.2321in,
width=3.2154in
]%
{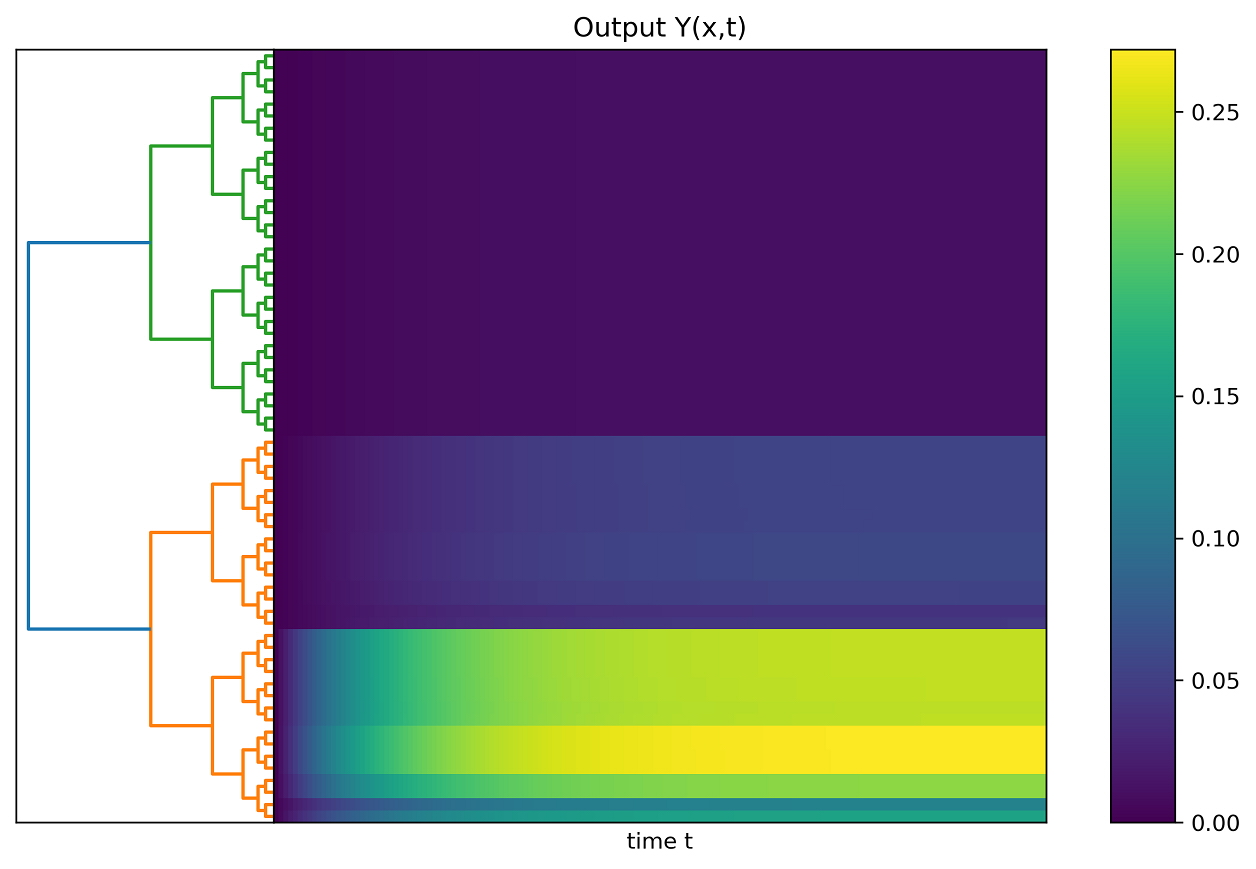}%
\caption{Simulation $3$. Step $0.05$.}%
\end{center}
\end{figure}

\newpage

\subsection{Fourth Simulation}

In this example, we take $k=2$, $p=2$, which means that we use a tree with
$2^{4}=16$ leaves and $4$ levels. The parameters of the CNN are $A(x)=\Omega
\left(  2^{2}\left\vert x-2^{-2}\right\vert _{2}\right)  $, $B(x)=U(x)=Z(x)=0$, we
set $X_{0}\left(  x\right)  =\Omega\left(  2^{2}\left\vert x\right\vert
_{2}\right)  $, $f(x)=\frac{1}{2}\left(  |x+1|-|x-1|\right)  $ for $x\in
G_{2}$.

In this example, at time zero the cells near the origin are excited. Which
causes all the cells of the network to activate. The activation can be seen in
the Fourier transform of the output. After some time the network returns to a
state of rest.%

\begin{figure}
[h]
\begin{center}
\includegraphics[
height=2.2321in,
width=3.6685in
]%
{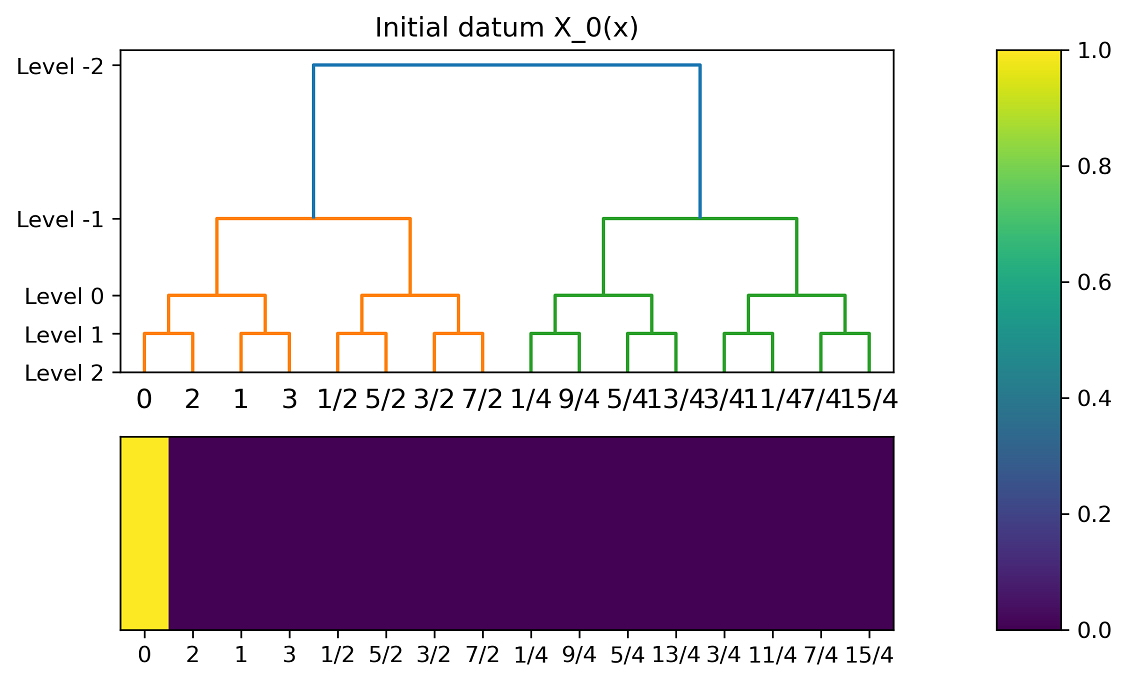}%
\caption{Simulation $4$. Heat map of $X_{0}(x)$.}%
\end{center}
\end{figure}
%

\begin{figure}
[h]
\begin{center}
\includegraphics[
height=2.2321in,
width=3.141in
]%
{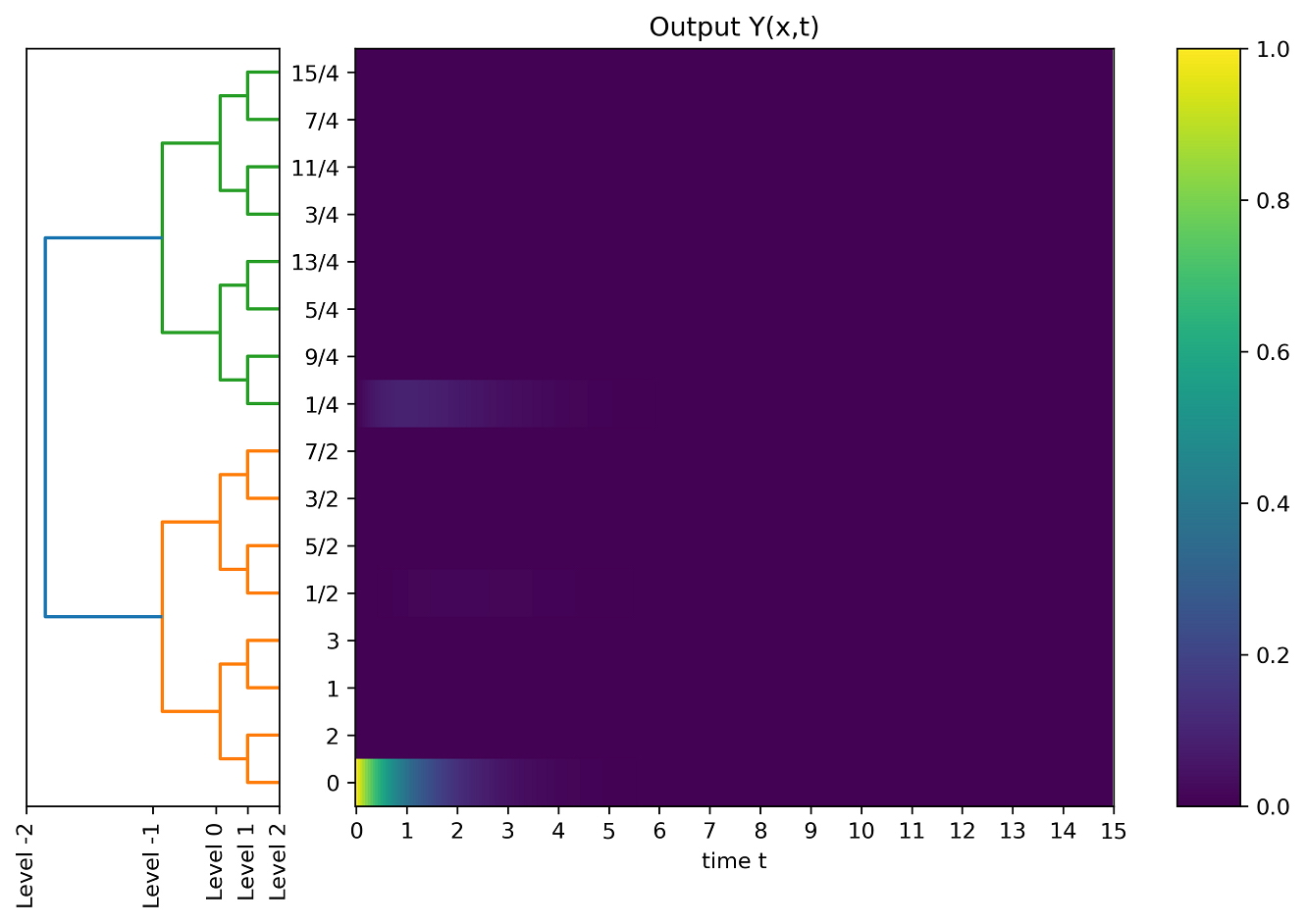}%
\caption{Simulation $4$. Step $0.05$.}%
\end{center}
\end{figure}
%

\begin{figure}
[h]
\begin{center}
\includegraphics[
height=2.2321in,
width=3.173in
]%
{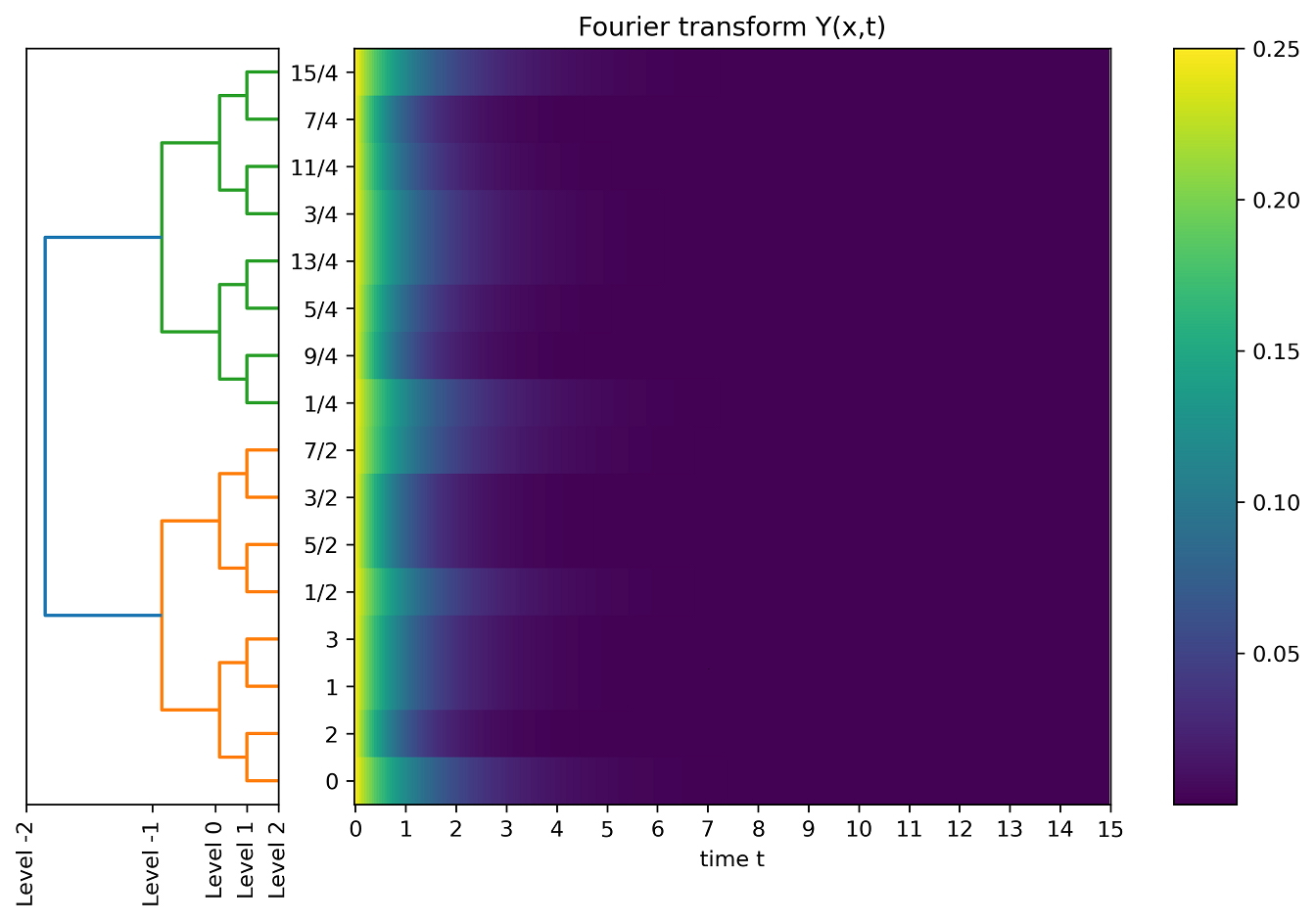}%
\caption{Simulation $4$.}%
\end{center}
\end{figure}

\newpage
\section{CONCLUSIONS}

In this article, we present a $p$-adic generalization of Chua-Yang CNNs. In
the $p$-adic framework, a continuous CNN is modeled by just one
integro-differential equation depending on several $p$-adic variables and the
time. In contrast, the classical CNNs are described by a discrete system of
integro-differential equations. The need of constructing continuous models of
discrete CNNs whose modeling requires millions of integro-differential
equations is quite natural.

A one-dimensional $p$-adic continuous CNN has infinitely many cells which are
hierarchically organized in rooted trees, also a such network has infinitely
many hidden layers. The topology of the network, which lately controls the
interaction of the cells, depends on the supports of the kernels of the
feedback and feedforward operators. Under mild hypotheses, \ there is a
natural discretization process of $p$-adic continuous CNNs that produces
standard discrete CNNs. The solutions of the continuous CNNs can be very well
approximated by the solutions of discrete CNNs. Then, for practical purposes,
a $p$-adic continuous CNN is a hierarchical discrete CNN with many hidden layers.

Our numerical simulations show that the solutions of continuous CNN exhibit a
very complex behavior, including self-similarity and multistability, depending
on the interaction of the parameters defining the network and the initial datum.

In the $p$-adic framework, the class of continuous CNNs is huge, for instance,
consider equations of type%

\[
\frac{\partial X(x,t)}{\partial t}=-\boldsymbol{L}X(x,t)+%
{\displaystyle\int\limits_{\mathbb{Q}_{p}^{N}}}
A(x,y)Y(y,t)d^{N}y+%
{\displaystyle\int\limits_{\mathbb{Q}_{p}^{N}}}
B(x,y)U(y)d^{N}y+Z(x),
\]
where $\frac{\partial X(x,t)}{\partial t}=-\boldsymbol{L}X(x,t)$ \ is a
$p$-adic heat equation, i.e. the fundamental solution of a such equation is
the transition probability density of a Markov process on $\mathbb{Q}_{p}^{N}%
$. The class of $p$-adic heat equations is extremely large, see e.g.
\cite{KKZuniga}, \cite{Zuniga-LNM-2016}. By incorporating a `diffusion term'
is natural to expect that the corresponding network will produce more complex
patterns. \ We plan to study these networks in a forthcoming article. In the
classical framework the reaction-diffusion CNNs have been studied
intennsively, see e.g. \cite{Gorals et al 1,Gorals et al 2,Gorals et al 3}.

\end{document}